\documentclass[12pt, draftclsnofoot, onecolumn]{IEEEtran}
\usepackage{amsthm}
\usepackage{amsmath}
\usepackage{algorithmic}
\usepackage[ruled]{algorithm2e}
\usepackage{mathrsfs}
\usepackage{cite}
\usepackage{bm,epsfig,amsthm,url}
\usepackage{indentfirst}
\usepackage{amssymb}
\usepackage{epstopdf}
\usepackage[caption=false,font=footnotesize]{subfig}
\usepackage{xcolor}
\usepackage{mathtools}

\theoremstyle{definition}

\newtheorem{proposition}{Proposition}

\newtheorem{remark}{Remark}

\DeclareMathOperator*{\argmin}{arg\,min}
\begin{document}
	
\title{UAV Aided Over-the-Air Computation}

\author{ Min~Fu,~\IEEEmembership{Student Member,~IEEE},~Yong~Zhou,~\IEEEmembership{Member,~IEEE},  Yuanming~Shi,~\IEEEmembership{Senior Member,~IEEE}, Wei~Chen,~\IEEEmembership{Senior Member,~IEEE}, and~Rui~Zhang,~\IEEEmembership{Fellow,~IEEE}%
	
	%\thanks{M. Fu is with the School of Information Science and Technology, ShanghaiTech University, Shanghai 201210, China, also with the Shanghai Institute
		%of Microsystem and Information Technology, Chinese Academy of Sciences,
		%Shanghai 200050, China, and also with the University of Chinese Academy
		%of Sciences, Beijing 100049, e-mail: fumin@shanghaitech.edu.cn}
	\thanks{M. Fu, Y. Zhou,  and Y. Shi are with School of Information Science and Technology, ShanghaiTech University, Shanghai 201210, China (e-mail: \{fumin, zhouyong, shiym\}@shanghaitech.edu.cn).} %
	\thanks{W. Chen is Department of Electronic Engineering and Beijing National Research Center for Information Science and Technology, Tsinghua University, Beijing 100084, China
	    (e-mail: wchen@tsinghua.edu.cn).}
    \thanks{ R. Zhang is with the Department of Electrical and Computer Engineering, National University of Singapore, Singapore 117583 (e-mail: elezhang@nus.edu.sg). }
  
	\thanks{This article has been presented in part at the \textit{IEEE Int.  Conf.  Commun.}, Montreal, Canada, Jun. 2021\cite{Fu2021UAV}.}
}

\maketitle
\setlength\abovedisplayskip{2pt}
\setlength\belowdisplayskip{2pt}
\vspace{-11mm}
\begin{abstract}
Over-the-air computation (AirComp) seamlessly integrates communication and computation by exploiting the waveform superposition property of multiple-access channels.  
Different from the existing works that focus on transceiver design of AirComp over static networks, this paper considers an unmanned aerial vehicle (UAV) aided AirComp system, where the UAV  as a flying base station aggregates data from mobile sensors. 
The trajectory design of the UAV provides an additional degree of freedom to improve the performance of AirComp.  
Our goal is to minimize the time-averaged mean-squared error (MSE) of AirComp by jointly optimizing the UAV trajectory, receive normalizing factors, and sensors' transmit power.	
To this end, we first propose a novel and equivalent problem transformation by introducing intermediate variables.  
This reformulation leads to a convex subproblem when fixing any other two blocks of variables, thereby enabling efficient algorithm design based on the principle of block coordinate descent and alternating direction method of multipliers (ADMM) techniques.
In particular, we derive the optimal closed-form solutions for normalizing factors and intermediate variables optimization subproblems.
We also recast the convex trajectory design subproblem into an ADMM form and obtain the closed-form expressions for each variable updating.
Simulation results show that the proposed algorithm achieves a smaller time-averaged MSE  while reducing the simulation time by orders of magnitude compared to state-of-the-art algorithms.
\end{abstract}
\vspace{-2mm}
\section{Introduction}\label{introduction}
The availability of massive sensory datasets and high-performance computing  platforms \cite{Savazzi2021FLMag} shall make connected intelligence a dominant feature of 6G wireless networks. 
Integrated sensing, communication, and computation is therefore required to enable a plethora of exciting data-intensive applications, including the internet of everything, tactile internet,  sustainable cities, and e-health.
Over-the-air computation (AirComp) \cite{Goldenbaum2013Robust, Goldenbaum2013Harnessing} is a disruptive technology that seamlessly integrates computation into communication, yielding a revolutionary paradigm shift from ``communicate then compute" to ``compute when communicate".
The basic principle of AirComp is to exploit the waveform/signal superposition property of  multiple-access channels (MAC) and to apply functional decomposition such that a base station (BS) directly obtains a class of nomographic functions of distributed data from concurrent sensor transmissions.
With the benefit of low-latency multiple access, AirComp has been recently applied to enable a wide range of internet of thing (IoT) applications, such as wireless federated machine learning \cite{Letaief2019Roadmap, Shi2020EdgeAI}, distributed consensus control \cite{Molinari2021Consensus}, and distributed sensing \cite{Zhu2020WDA}.

To enable reliable AirComp, one key aspect is the joint transceiver design (e.g., transmit power control, receive normalizing factor setting, and receive beamforming design) to reduce the computation error induced by the receiver noise and non-uniform channel fading. 
Therein, for  single-input-single-output (SISO) AirComp, the authors in  \cite{Liu2020Scaling, Cao2020Optimized}   proposed a computation-optimal policy  to  balance the trade-off between the noise-induced error and the signal-misalignment error.
In particular, the optimal transmit power control policy is shown to be  a combination of channel inversion policy and full power policy with a threshold-based structure. 
Moreover, in most prior works on multiple-input-single-output (MISO) AirComp \cite{Chen2018UniformForcing, Yang2020FL, Li2019Wirelessly} and multiple-input-multiple-output (MIMO) AirComp \cite{Li2019Wirelessly, Chen2018IoT, Wen2019Reduced, Zhu2019Aircompmobility}, the zero-forcing policy is commonly adopted for transmit power and receive normalizing factor  control. 
This policy perfectly compensates for the magnitude attenuation of signals at the expense of increasing the noise-induced error.
Unfortunately, when one or more individual channels are in deep fading, this policy may magnify the negative impact of noise and degrade the AirComp performance.
This is because the computation error is negatively correlated with the worst channel gain among all sensors.
Furthermore, to enhance the performance of AirComp, the authors in \cite{Jiang2019RIS, Wang2020RIS}  adopted an emerging reconfigurable intelligent surface (RIS) technology \cite{Yuan2021RIS,Fu2021Intelligent} to 
further optimize the passive beamforming at the RIS together with the transceiver design.

Most of the existing studies on AirComp \cite{Chen2018UniformForcing, Yang2020FL, Li2019Wirelessly, Chen2018IoT, Wen2019Reduced, Zhu2019Aircompmobility, Liu2020Scaling, Cao2020Optimized, Jiang2019RIS, Wang2020RIS} were restricted to static networks, where the positions of the sensory devices remain unchanged during the data aggregation process. 
However,  in some emerging  applications (e.g., consensus control \cite{Liu2020Scaling} and sensing \cite{Zhu2019Aircompmobility}), the sensors are usually embedded in  mobile devices (e.g., ground vehicles)
and may move out of the coverage area of the static ground BSs.
For example,  sensors can be mounted on the mobile ground vehicles to monitor a wild environment to avoid natural disasters \cite{Zhu2019Aircompmobility}.
Under these circumstances, the performance of AirComp may be severely degraded, especially 
in remote areas,  where the ground BSs are usually sparsely deployed or unavailable.
%In addition, due to the analog transmission, the performance of AirComp is vulnerable to receiver noise, especially in long-distance communications.
Moreover, because of  the channel fading, the receiver noise, and the limited transmit power at the sensors, only relying on the  transceiver design is not able to guarantee the performance of AirComp.
Therefore, it is necessary to deploy a more flexible BS to deal with the aforementioned challenges.

Fortunately, low-cost unmanned aerial vehicles (UAVs)  are considered as a promising alternative to assist the terrestrial networks \cite{Zeng2019Accessing, Zhang2020Energy}.
Recently, a growing body of research efforts have been devoted to study the deployment of UAVs as mobile BSs in  IoT networks to support information dissemination \cite{Zeng2017EnergyEfficient, Wu2018MultiUAV} and  data collection \cite{Zhan2019Time}.
This motivates us to deploy a UAV as a flying BS to aggregate data from mobile sensors via  AirComp in IoT networks, where the ground sensors  are continuously moving  and  the terrestrial BS is unavailable.  
UAV-aided AirComp enjoys the following advantages.  
First, the UAV-mounted BS is cost-effective and can be flexibly deployed to provide services when the terrestrial BS is not available.
Second, due to UAV's high altitude, the UAV can establish line-of-sight (LoS) connections with the sensors to  alleviate the performance loss of AirComp induced by channel impairments (e.g., fading).
Finally,  with controllable mobility,  the UAV can track sensors' movement to avoid long-distance transmissions and dynamically strike a balance between communication distance and sensors' transmit power, thereby enhancing the performance of AirComp.
%to dynamically strike a balance between communication distance and sensors' transmit power, thereby enhancing the performance of AirComp.

To quantify the computation error, a performance metric that has been widely adopted for AirComp is the mean-squared error (MSE) between the estimated function value and the target function value \cite{Chen2018UniformForcing, Yang2020FL, Li2019Wirelessly, Chen2018IoT, Wen2019Reduced, Zhu2019Aircompmobility, Liu2020Scaling, Cao2020Optimized, Jiang2019RIS, Wang2020RIS}.
Hence, we formulate a time-averaged MSE minimization problem with the joint UAV trajectory and transceiver design, taking the maximum speed at the UAV as well as the peak and average transmit power at the sensors into consideration.
We aim to optimally balance the trade-off between the signal-misalignment error and the noise-induced error.
However, due to the highly coupled variables and time-dependent constraints, it is generally challenging to solve the formulated problem optimally.
The conventional method to decouple the variables is the block coordinate descent (BCD) method \cite{xu2013block}, which updates each block in an alternating manner until convergence, resulting in a non-convex trajectory design subproblem.
Although the successive convex approximation (SCA) method \cite{marks1978general} can tackle the non-convex trajectory design subproblem with the first-order Taylor approximation, it is not guaranteed to find an optimal solution of the original subproblem.
Furthermore, by using  CVX and interior-point solvers (e.g., SDPT3)\cite{Grant2014CVX, Shi2015Largescale}, the SCA method is time-consuming and may not be scalable for large-scale networks. 
In contrast, the alternating direction method of multipliers (ADMM) \cite{Bertsekas1997} is a powerful first-order method that is well-suited for large-scale convex optimization.
As inspired,  we in this paper develop a novel problem reformulation that leads to a convex trajectory design subproblem, and propose an ADMM method to solve the aforementioned subproblem.

\subsection{Contributions}
The main contributions of this paper are summarized as follows.
\begin{itemize}
\item 
This paper is one of the early attempts to study the UAV-aided AirComp system, where the UAV is deployed to provide services when the terrestrial BS is unavailable and track sensors' mobility for establishing LoS connections, thereby enhancing the performance of AirComp and the robustness against noise.
Additionally, we formulate a time-averaged MSE minimization problem by jointly designing the UAV trajectory, normalizing factors at the UAV, and transmit power at the sensors, taking into account the maximum speed of the UAV, as well as the peak and average transmit power budgets at the sensors.

\item  
To address the limitations of the existing methods, we first introduce intermediate variables as alternatives for transmit power variables,  termed as {\it{signal quality factors}} defined as the product of each sensor's transmit power and its channel power gain.
Although the reformulated problem is still non-convex, this novel variable transformation makes it easier to  decompose the original problem into three convex subproblems, which can be optimally solved.

\item 
Based on the above results, we  exploit the BCD method  to decompose the reformulated problem.
We derive the optimal closed-form solutions  for the receive normalizing factors and signal quality factors. 
To further reduce the computational complexity,  we rewrite the convex trajectory  optimization subproblem
as an ADMM form, which can update each variable in closed form.

\end{itemize}

The numerical results validate that the importance and necessity of the joint UAV trajectory and transceiver design for minimizing the time-averaged MSE and enhancing the robustness against noise in mobile networks.
It also shows that the proposed algorithm achieves significantly performance gains and reduces the simulation time by orders of magnitude compared to the existing algorithms.

\subsection{Organization}
The remainder of this paper is organized as follows. Section \ref{model}
 describes the system model and problem formulation. 
In Section \ref{Algorithm section}, we propose a BCD-ADMM method to solve the formulated problem.
Section \ref{NR section} presents the numerical results to evaluate the performance of the proposed algorithm. Finally, we conclude this paper in Section \ref{C section}.

\section{System  Model and Problem Formulation}\label{model}
\begin{figure}[t]
	\centering
	\includegraphics[width= 0.5\linewidth]{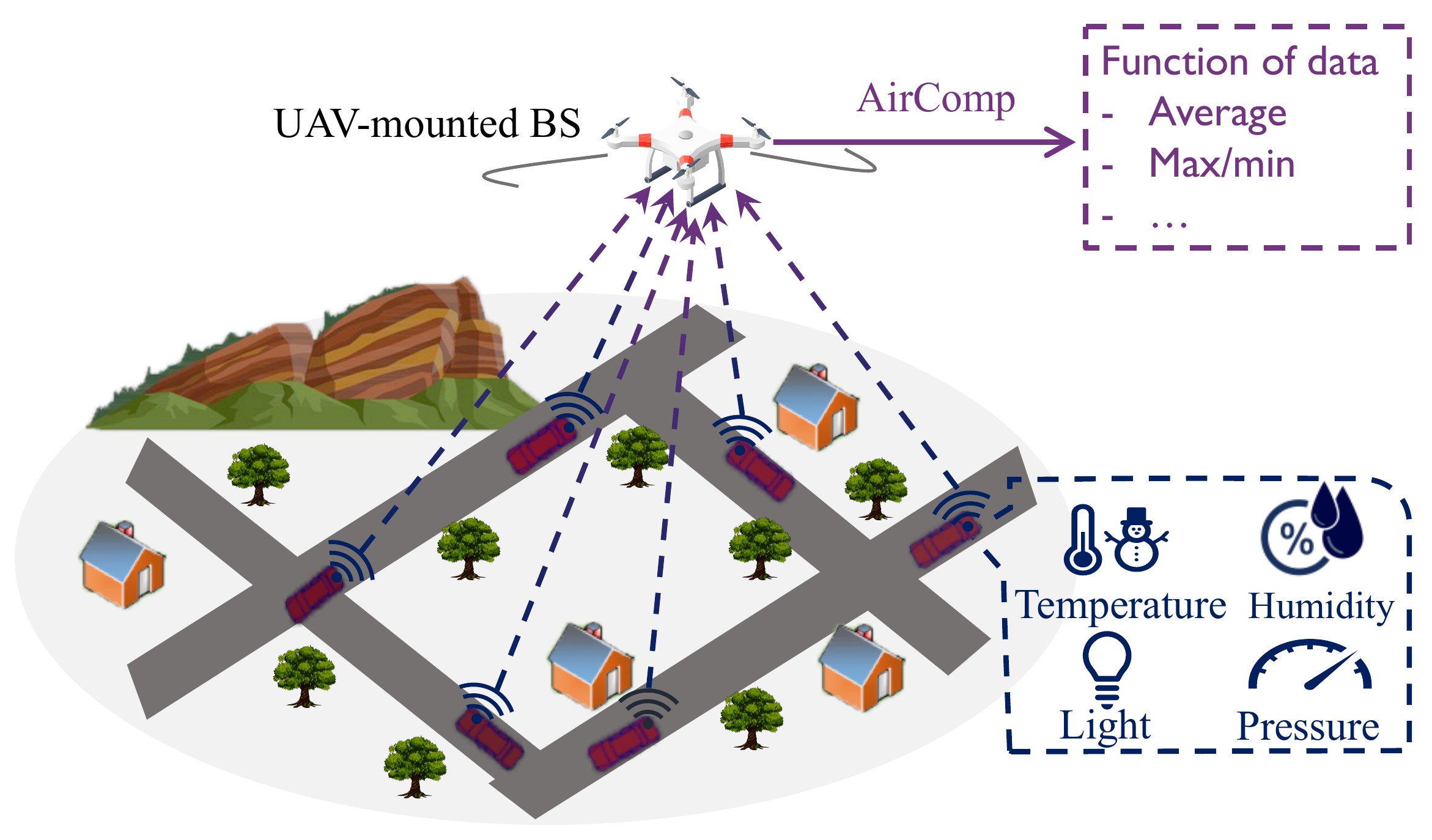}
		\vspace{-4mm}
	\caption{An illustration of a UAV-aided AirComp network.}\label{F:Systemmodel}
	\vspace{-7mm}
\end{figure}
As illustrated in Fig. \ref{F:Systemmodel},  we consider a UAV-aided AirComp network with $K$ ground mobile sensors, where the terrestrial BS is not available.
Therein, the UAV is deployed as an aerial BS to aggregate sensory data from $K$ distributed sensors during a given mission duration of  $T$ second (s). 
In addition, we assume that each ground sensor moves with a given speed and along a path designed in advance to collect data at different locations \cite{Zhu2019Aircompmobility, Zhang2021UAVmobility, Liu2020UAV-PM}. 
One practical scenario for such a consideration could be wild-area environmental monitoring \cite{Zhu2019Aircompmobility}, where the BSs are unavailable nearby and the UAV is dispatched to monitor the average temperature/humidity measured by the sensors, while the  sensors mounted on moving vehicles are employed to measure temperature/humidity data at different positions.
Both the UAV and the sensors are equipped with a single antenna due to their size and power limitations. 
To achieve ultrafast data aggregation, the  UAV  exploits AirComp to aggregate a nomographic function (e.g., average function and sum function) of the distributed data, rather than to decode each sensory data separately.

%In an effort to provide the enhanced quality of service, UAV has to move according to the prediction of the movement of the sensors for establishing LoS connections and adaptively balance between communication distance and sensors' transmit power, thereby saving sensors'  power for the signal magnitude alignment and suppress noise.
%Note that, for the random mobility of sensors, the positions of sensors can be predicted in advance using long-short term memory or echo state networks based on real dataset\cite{Chen2017UAVmobility},\cite{Liu2019UAVmobility}.

\subsection{UAV Mobility Model and Channel Model}
\subsubsection{UAV Mobility Model} 
Since the sensors are moving continuously,  we aim to design the UAV trajectory according to the sensors' locations so as to reduce the computation error.
The mobile UAV can move sufficiently close to the sensors for avoiding long-distance transmissions, thereby saving sensors' transmit power and mitigating the effect of noise.
%Furthermore,  the UAV can dynamically adjust its trajectory according to sensors' transmit power to balance the communication distances between it and sensors such that sensors' effective channel power gains are similar, which helps to align signals.
%Intuitively,  the UAV moves closer to sensors with lower power budgets compared to those with higher power budgets.
Therefore, the UAV trajectory design provides an additional degree of freedom for  AirComp performance enhancement.
In a three-dimensions (3D) Cartesian coordinate
system, we denote the location of the UAV at time $t$ projected on the horizontal (ground) plane as $\bm{q}(t)=[x(t), y(t)]\in \mathbb{R}^{1\times 2}$, $0\leq t\leq T$ with $x(t)$ and $y(t)$ being  $x$- and $y$-coordinates at time instant $t$, respectively.
We assume that the UAV flies at a fixed altitude $H$ above the ground level.
Note that in practice, $H$ corresponds to the minimum altitude
that ensures obstacle avoidance without the need for frequent
aircraft ascending and descending.
In addition, the UAV is assumed to start the mission at an initial location, the horizontal coordinate of which is denoted as $\bm{q}_{I} = [x_I, y_I]\in \mathbb{R}^{1\times 2}$. 
Note that the initial location is determined by various factors, e.g., energy replenishment \cite{Zhan2019Time}, \cite{Group2020MultiUAV}.
 We denote $\bm q(0) = \bm q_{I}$. 
 Besides, we denote the  maximum speed of the UAV as $V_{\rm{max}}$ in meter/second (m/s).
 Hence, we have the constraints
$\sqrt{\dot{x}^2(t) +\dot{y}^2(t)} \leq V_{\text{max}}$, $0 < t < T$, where
$\dot{x}(t)$ and $\dot{y}(t)$ denote the time-derivatives of $x(t)$ and $y(t)$ at time instant $t$, respectively.

To assist a tractable algorithm design, we adopt the time discretization technique to deal with the continuous UAV trajectory design, which is widely considered in most of the existing works \cite{Wu2018MultiUAV, Zeng2017EnergyEfficient, Zhan2019Time, Group2020MultiUAV, Tang2019Alternative}.
 Specifically, the mission duration $T$ is equally divided into $N$ time slots, i.e., $ T = N\delta $, where $\delta$ denotes  time step size.
Given the maximum UAV speed $V_{\max}$ and altitude $H$,  the time step size $\delta$ needs to satisfy $\delta V_{\text{max}}\ll H$ so that the distance between the UAV and  sensors is approximately a constant during each time slot.
Based on the time discretization technique, the UAV trajectory $\bm q(t)$ over time horizon $T$ is approximated by the $(N+1)$-length sequence $\{\bm{q}[n]\}_{n =0}^{N}$ with $\bm{q}[n] \triangleq \bm q(n\delta) $ denoting the UAV's horizontal coordinate at time slot $n$. 
We denote  $\mathcal {N}$ as $\mathcal {N}\triangleq \{1,\ldots,N\} $.
 The UAV's mobility constraints are given by
 \setlength\arraycolsep{2pt}
\begin{eqnarray}
&&\|\bm q [n]-\bm q[n-1]\|_2 \leq V_{\max}\delta,  \forall n \in \mathcal {N},  \label{speed constr}\\
&&\bm{q}[0] = \bm q_{I},  \label{initial location}
\end{eqnarray} 
where constraints \eqref{speed constr} correspond to the UAV speed constraint and constraint \eqref{initial location}  is subject to the initial location of the UAV.

\subsubsection{Channel Model}
Recent field experiments by Qualcomm verified that UAV-to-ground channels are  dominated by  LoS links when the UAV flies above a certain altitude \cite{Qual2017LET}. 
In this paper, we assume that each connection from the sensor to the UAV is dominated by the LoS channel.
Moreover, the Doppler effect resulting from  mobility is assumed to be perfectly compensated \cite{Mengali1997Sync}.
The horizontal coordinate of sensor $k$ at time slot $n$ is denoted as  $\mathbf w_k[n]=[x_k[n], y_k[n]]\in \mathbb{R}^{1 \times 2}$ with $x_k[n]$ and $y_k[n]$ being  $x$- and $y$-coordinates at time slot $n$, respectively. The set of ground sensors is denoted as  $\mathcal {K}\triangleq \{1,\ldots,K\} $, $K> 1$.
The time-varying channel from sensor $k$ to  the UAV at time slot $n$ follows the free-space path loss model \cite{Wu2018MultiUAV}
\begin{eqnarray}
{h}_k[n]=\sqrt{\beta_k[n]}\tilde{h}_k[n], \forall k \in \mathcal {K}, \label{chan0}
\end{eqnarray}
where $|\tilde{h}_k[n]| =1$,  and $\beta_k[n]$ denotes the free-space path loss. Specifically, $\beta_k[n]$ is modeled as $\beta_k[n]=\beta_0 d_k^{-2}[n],$
where $\beta_0$ represents the channel power gain at the reference distance of 1 m related to the carrier frequency and antenna gain, and $d_k[n] = \sqrt{H^2+\|\bm{q}[n]-\textbf{w}_k[n]\|_2^2}$ is the  distance between the UAV and sensor $k$ at time slot $n$.

\subsection{AirComp for Data Aggregation}
Let $z_{k}[n] \in \mathbb{C}$ denote the sensory data of sensor $k$ at  time slot $n$.
The UAV intends to obtain a  function (e.g., geometric mean and arithmetic mean) of the measured data from $K$  sensors at each time slot $n\in{\cal N}$, where the target function of $K$ variables is denoted as $f[n]:\mathbb{C}^{K}\rightarrow\mathbb{C}$.
By using a mathematical property of theoretical function representation, the target function can be expressed as its nomographic form as a function of a finite sum of univariate functions \cite{Liu2020Scaling}. 
Specifically, the target function computed at the UAV is written as in its nomographic form:
\begin{eqnarray}
f[n]\big(z_{1}[n],\ldots, z_K[n]\big) =\phi \bigg(\sum \limits_{k=1}^{K}\psi_k\big(z_{k}[n]\big)\bigg),\label{functions}	
\end{eqnarray}
where $\psi_k:\mathbb{C}\rightarrow\mathbb{C}$ is the pre-processing function, $\forall k\in \mathcal{K}$ and $\phi:\mathbb{C}\rightarrow\mathbb{C}$ is the post-processing function.
Note that the UAV's original computation of $f[n]$ by processing $K$ variables has been decomposed into $(K+1)$ small tasks of designing $\{\psi_1,\ldots, \psi_K, \phi\}$.
Based on the above function representation, we present an efficient AirComp technique \cite{Liu2020Scaling, Cao2020Optimized} for the low-latency target function $f[n]$ computation by exploiting the waveform/signal superposition property of MACs.
With AirComp,  each sensor pre-processes its own signal $z_{k}[n]$ with function $\psi_k$  and simultaneously transmits $\psi_k(z_{k}[n])$ to the UAV, while the UAV post-processes the received sum of signals $\sum_{k=1}^{K}\psi_k\big(z_{k}[n]\big)$ with function $\phi$ to estimate the desired computation $f[n]$. 
The sensors' transmissions are assumed to be well synchronized \cite{Liu2020Scaling, Cao2020Optimized}.

Without loss of generality, in this paper, we consider the case where the UAV  computes the average of distributed data generated by the sensors \cite{Yang2020FL}, \cite{Cao2020Optimized}.
Therefore, the   function of interest at the UAV at time slot $n$ is given by
\vspace{-2mm}
\begin{eqnarray}\label{average}
f[n]=\frac{1}{K}\sum \limits_{k=1}^{K}\psi_k\big(z_{k}[n]\big).
\end{eqnarray}
To compute  function $f[n]$ in \eqref{average} via AirComp, the specific procedure is described as follows.
The transmit signals after pre-processing at each sensor are give by
\begin{eqnarray}
s_k[n] = \psi_k(z_{k}[n]), \forall k\in \mathcal{K},
\end{eqnarray}
where  $\{s_{k}[n],\forall k\in \mathcal{K}\}$ are assumed to be independent with each other and normalized with zero mean and unit variance, i.e., $\mathbb{E}(s_k[n]) = 0$, $\mathbb{E}(s_k[n]s_k^{\sf H}[n]) = 1$, and $\mathbb{E}[s_{i}[n]s_{j}[n]^{\sf H}] = 0 , \forall i \neq j$, as in \cite{Yang2020FL}, \cite{Cao2020Optimized}.
After the sensors simultaneously send their pre-processed signals $\{s_{k}[n]\}$ to the UAV over a single frequency channel, the received signal at the UAV is given by
\begin{eqnarray}
y[n]=\sum \limits_{k=1}^{K}b_k[n]h_{k}[n]s_{k}[n]+ e[n],\label{time_Y}
\end{eqnarray}
where $b_k[n] \in \mathbb{C}$ denotes the transmit pre-coding coefficient at sensor $k$ for channel-fading compensation and $e[n]$ denotes the  additive white Gaussian noise (AWGN), i.e., $e[n]\thicksim\cal{C}\cal{N}$$(0,\sigma^2)$.  
Therein, the peak and average transmit power constraints at sensor $k$ are respectively given by
\begin{eqnarray}
	|b_k[n]|^2 \leq  P_k, \ \frac{1}{N}\sum_{n=1}^{N} |b_k[n]|^2 \leq  \bar{P}_k.\label{average p}
\end{eqnarray}

Upon receiving  signal $y[n]$ in \eqref{time_Y}, the estimated average function after post-precessing at the UAV is given by
\vspace{-3mm}
\begin{eqnarray}\label{ave_functions}
\hat f[n] = \frac{y[n]}{K\eta[n]},
\end{eqnarray}
where  $\eta[n]\in \mathbb{C}$ is a receive normalizing factor at the UAV. Note that it is applied to both signals and the noise, and is designed to provide power compensation for signals and suppress the noise, thereby obtaining an accurate estimation of the target function $f[n]$.

\subsection{Performance Metric}
To quantify the AirComp performance, the distortion of the estimation function $\hat f[n]$ with respect to (w.r.t.) the desired function $f[n]$ is measured by the MSE between $\hat f[n]$ and   $f[n]$, which is widely adopted  in the existing AirComp studies \cite{Chen2018UniformForcing, Yang2020FL, Li2019Wirelessly, Chen2018IoT, Wen2019Reduced, Zhu2019Aircompmobility, Liu2020Scaling, Cao2020Optimized, Jiang2019RIS, Wang2020RIS}.
In particular, the corresponding  MSE at time slot $n$ is given by
\begin{eqnarray}\label{MSE_n}
\hspace{-2em}{\sf MSE}[n] = \mathbb{E}[|\hat{f}[n]-f[n]|^2]
&=&\frac{1}{K^2}\mathbb{E}\Bigg[\Bigg(\frac{y[n]}{\sqrt{\eta[n]}}- \sum \limits_{k=1}^Ks_{k}[n]\Bigg)^2 \Bigg] \nonumber \\
&=&\frac{1}{K^2}\Bigg[\sum \limits_{k=1}^{K}\Bigg(\frac{b_k[n]h_{k}[n]}{\eta[n]}-1\Bigg)^2+\frac{\sigma^2}{|\eta[n]|^2}\Bigg],
\end{eqnarray}
where the expectation is taken over the distributions of the transmitted signals $\{s_k[n]\}$ and  noise $e[n]$.
Accordingly, the time-averaged MSE over $N$ time slots is given as
\begin{eqnarray}\label{MSE_s_y_t}
\overline{\sf MSE} =  \frac{1}{N}\sum\limits_{n=1}^{N} {\sf MSE}[n].
\end{eqnarray}

In this paper, we aim to minimize the time-averaged MSE by jointly optimizing  $\{b_k[n]\}$,  $\{\eta[n]\}$, and $\{\bm q[n]\}$. 
Prior to formulating the optimization problem, we  present some properties of objective  function \eqref{MSE_s_y_t} as follows. 
\begin{proposition}\label{real proposition}
With given any  amplitudes of complex transmit pre-coding coefficients $\{b_k[n]\}$, complex normalizing factors $\{\eta[n]\}$, and complex channel coefficient $\{h_k[n]\}$, to achieve the minimum $\overline{\sf MSE}$, each term $b_k[n]h_{k}[n]/{\eta[n]}$ in  \eqref{MSE_s_y_t} must be  real and non-negative for all $n\in\mathcal{N}, k\in\mathcal{K}$.
\end{proposition}
\begin{proof}
	Please refer to Appendix \ref{real proposition proof}.
\end{proof}
Based on \textbf{Proposition} \ref{real proposition}, without loss of generality,
we set $\eta[n]\in \mathbb{R}^+$, and $b_k[n]\triangleq \frac{ \sqrt{p_{k}[n]} h_{k}^{\dagger}[n]}{ |h_{k}[n]|}$ with $p_{k}[n]\in [0, P_k]$ to perfectly offset the phases introduced by the complex channel coefficients such that  each term ${b_k[n]h_{k}[n]}/{\eta[n]}$ in \eqref{MSE_s_y_t} is real and non-negative, in the rest of paper. 
In this sense, it  allows us to focus on the control of $p_k[n]$ instead of $b_k[n]$, where  $p_k[n]$ represents the transmit power at sensor $k$ at  time slot $n$.
Hence, the time-averaged MSE is rewritten as 
\begin{eqnarray}\label{MSE_s_y_t1}
\hspace{-2em}\overline{\sf MSE}  &=&  \frac{1}{NK^2}\sum_{n=1}^{N}\Bigg[\sum \limits_{k=1}^{K}\Bigg(\frac{\sqrt{p_{k}[n]}\big|h_k[n]\big|}{\eta[n]}-1\Bigg)^2 +\frac{\sigma^2}{\eta^2[n]}\Bigg] \nonumber \\
&=&\frac{1}{NK^2}\sum_{n=1}^{N}\Bigg[\sum\limits_{k=1}^{K}\Bigg(\frac{\sqrt{p_{k}[n]}\sqrt{\beta_0}}{\eta[n]\sqrt{(H^2+\|\bm{q}[n]-\textbf{w}_k[n]\|_2^2)}}-1\Bigg)^2 +\frac{\sigma^2}{\eta^2[n]}\Bigg].
\end{eqnarray}

\subsection{Problem Formulation}
Let  $\bm p = \{ p_{k}[n],\forall n\in\mathcal{N},\forall k\in\mathcal{K}\}$, $\bm \eta = \{\eta[n],\forall n\in\mathcal{N}\}$, and $\bm q = \{\bm q[n],\forall n = 0, \ldots, N \}$.
By assuming that the time-dependent locations of the ground sensors are known in advance \cite{Liu2020UAV-PM}, our goal is to minimize $\overline{\sf MSE} $ by jointly optimizing the transmit power $\bm p$ of the sensors, the normalizing factors $\bm \eta$  at the UAV, and the UAV trajectory $\bm q$ over different time slots. 
The optimization problem  is formulated as 
\begin{subequations}\label{original formulation}
\begin{eqnarray}
\hspace{-6mm}\mathop{\text{minimize}}_{\bm p, \bm \eta, \bm q } &&  \overline{\sf MSE}\nonumber \\
\text{subject to} 
 && 0\leq p_k[n]\leq P_k,  \forall k \in \mathcal {K},  \forall n \in \mathcal {N}, \label{p1001} \\
 &&0\leq\frac{1}{N}\sum_{n=1}^N p_k[n]\leq \bar{P}_k,  \forall k \in \mathcal {K}, \label{p1002}\\
 &&\eta[n] \geq 0,  \forall n \in \mathcal {N},\label{etacons} \\
&& \|\bm q [n]-\bm q[n-1]\|_2 \leq V_{\max}\delta,\forall n \in \mathcal {N},\label{qcons: speed}\\
&&\bm{q}[0] = \bm q_{I}.  \label{qcons: initial}
\end{eqnarray}
\end{subequations}
To make constraint \eqref{p1002} non-trivial, we assume $\bar{P}_k<P_k$ in this paper.
Note that the challenges of solving problem \eqref{original formulation} lie in the following two main aspects. 
First, the transmit power  $\bm p$, normalizing factors $\bm \eta$, and UAV trajectory $\bm q$ are highly coupled over different time slots. 
Second, for fixed  transmit power  $\bm p$ and normalizing factors $\bm \eta$, although all the constraints  of problem \eqref{original formulation} are convex w.r.t. $\bm q$, the objective function $\overline{\sf MSE}$ is still non-convex w.r.t. $\bm q$. 
As a result, problem \eqref{original formulation} is a non-convex optimization problem.
In general, there is no standard method for solving such non-convex optimization problems optimally. 

The BCD-SCA method proposed in\cite{Fu2021UAV} can be adopted to solve problem \eqref{original formulation}.
Specifically,  the variables $\bm p$, $\bm \eta$, and $\bm q$ are decoupled by exploiting the BCD method.
And the non-convexity of objective function $\overline{\sf MSE}$ in the resulting subproblem related to $\bm q$ is  tackled by adopting the SCA technique. 
However, the SCA algorithm only optimizes the approximate lower bound of trajectory subproblem. 
Consequently, the  aforementioned approximate algorithm is not guaranteed to find an optimal solution of the resulting subproblem related to $\bm q$, which may result in low-quality solutions.  
Moreover, the BCD-SCA method in \cite{Fu2021UAV} always relies on the CVX and interior-point solvers (e.g., SDPT3) to solve the approximated convex trajectory optimization subproblems numerically.
Clearly, the computational costs of these second-order algorithms are not scalable when the problem size ($N$ or $K$) is large.

 To address the limitations of the existing methods, we first transform problem \eqref{original formulation} into an equivalent and more tractable form in the sequel. 
Specifically, for each sensor $k$, we define \textit{signal quality factor} at each time slot $n$ as the product of its transmit power and channel gain (i.e., $\theta_k[n] \triangleq p_k[n]|h_k[n]|^2$).
 Let $\bm \theta = \{\theta_k[n], \forall k \in \mathcal{K}, \forall n \in \mathcal{N}\}$.
By introducing $\theta_k[n]$,  problem \eqref{original formulation} can be equivalently transformed as 
\begin{subequations}
\begin{eqnarray}
\mathscr{P}:\mathop{\mathop{\text{minimize}}_{\bm \theta, \bm \eta, \bm q }}   && \frac{1}{NK^2}\sum_{n=1}^{N}\Bigg[\sum_{k=1}^K\!\Bigg(\frac{\sqrt{\theta_{k}[n]}}{\eta[n]}-1\Bigg)^2+\frac{\sigma^2}{\eta^2[n]}\Bigg] \nonumber \\
\text{subject to}
&&0\leq \frac{\theta_{k}[n]}{|h_k[n]|^2}\leq P_k,  \forall k,  \forall n, \label{thetacons: peakP} \\
 &&0\leq\frac{1}{N}\sum_{n=1}^N \frac{\theta_{k}[n]}{|h_k[n]|^2}\leq \bar{P}_k,  \forall k, \label{thetacons: averageP}\\
 && \text{constraints}\ \eqref{etacons}, \eqref{qcons: speed}, \eqref{qcons: initial}.
\end{eqnarray}
\end{subequations}
It is easily verified that problem $\mathscr{P}$ is equivalent to problem \eqref{original formulation}.
With such a transformation, we only need to focus on solving problem $\mathscr{P}$ in the rest of the paper.
Although problem $\mathscr{P}$ is still a non-convex optimization problem due to the coupled optimization variables, it reduces to a convex subproblem when any two blocks of variables $\bm \theta$, $\bm \eta$, and $\bm q$ are fixed, which is not the case for problem \eqref{original formulation}.
This facilitates the development of an efficient algorithm with high-quality solutions shown later in Section \ref{Algorithm section}.
\begin{remark}
One can observe that the objective function of problem $\mathscr{P}$ consists of two components, which are the signal misalignment error (i.e., $\frac{1}{NK^2}\sum_{n=1}^{N}\sum_{k=1}^K(\frac{\sqrt{\theta_{k}[n]}}{\eta[n]}-1)^2$) and the noise-induced error (i.e., $\frac{1}{NK^2}\sum_{n=1}^{N}\frac{\sigma^2}{\eta^2[n]}$).
Ideally, to minimize the time-averaged MSE, we can enlarge $\bm \eta$ to suppress the noise-induced error while compelling the signal quality factors $\bm \theta$ to equal the normalizing factors to perfectly align the signals.
Unfortunately, due to the limited power budget and non-uniform channel fading, the values of the signal quality factors are usually limited.
Thanks to the UAV's mobility, we can design the UAV trajectory according to the sensors' locations to construct favorable channels and thus increase the signal quality factors compared to the static UAV/BS case, thereby decreasing the time-averaged MSE.
\end{remark}

%%%%%%%%%%%%%%%%%%%%%%%%%%%%%%%%%%%%%%%%%%%%
\section{BCD-ADMM Method for Solving Problem $\mathscr{P}$} \label{Algorithm section}
In this section, we develop a low-complexity algorithm, named  BCD-ADMM method, to solve problem $\mathscr{P}$ efficiently.
Specifically, to address the highly coupled optimization variables, we adopt the BCD method to decompose the joint
optimization problem $\mathscr{P}$ into three convex quadratically constrained quadratic programming (QCQP) subproblems, whose global optimal solutions can be obtained. 
To further reduce the computational complexity, we derive closed-form expressions for $\bm \theta$ and $\bm \eta$, followed by presenting a low-complexity implementation using the ADMM to solve the convex trajectory optimization subproblem with closed-form solutions for each variable updating.  
\subsection{Normalizing Factors Optimization}
In this subsection, given $\bm \theta$ and $\bm q$, we reformulate problem $\mathscr{P}$ by optimizing  $\bm \eta$  as
	\begin{eqnarray}
	\mathscr{P}_{1.1}:\mathop{\text{minimize}}_{\{\eta[n]\geq 0\}}  &&\sum_{n=1}^{N}\Bigg[\sum \limits_{k=1}^K\Bigg(\frac{\sqrt{\theta_k[n]}}{\eta[n]}-1\Bigg)^2+\frac{\sigma^2}{\eta^2[n]}\Bigg]. \nonumber
	\end{eqnarray}
Problem $\mathscr{P}_{1.1}$  can be decoupled into $N$ subproblems each for optimization $\eta[n]$ to minimize the MSE. The $n$-th subproblem is written as
\begin{eqnarray}\label{n-th eta}
\mathop{\text{minimize}}_{\eta[n]\geq 0} &&\sum \limits_{k=1}^K\Bigg(\frac{\sqrt{\theta_{k}[n]}}{\eta[n]}-1\Bigg)^2+\frac{\sigma^2}{\eta^2[n]}\!.
\end{eqnarray}
By denoting $\nu[n] = 1/\eta[n]$, problem $\eqref{n-th eta}$ can be transformed to a convex quadratic problem as
\begin{eqnarray}\label{n-th nu}
\mathop{\text{minimize}}_{\nu[n]\geq 0}  && \sum \limits_{k=1}^K\left(\sqrt{\theta_k[n]}\nu[n]-1\right)^2+\sigma^2(\nu[n])^2.
\end{eqnarray}

By setting the first derivative of the objective function in problem $\eqref{n-th nu}$ to be zero, we can obtain the optimal solution $\nu^{\star}[n]$ to problem \eqref{n-th nu}. As a result, the optimal solution to problem \eqref{n-th eta} is obtained as $\eta^{\star}[n] = {1}/{\nu^{\star}[n]}, \forall n$, given in the following proposition.
\begin{proposition}
With any given $\bm \theta$ and $\bm q$, the optimal solution  $\bm \eta$ of problem  $\mathscr{P}_{1.1}$ is given by
\begin{eqnarray} \label{solution denoise}
\eta^{\star}[n]  =  \frac{\sigma^2+ \sum_{k=1}^K\theta_{k}[n]}{\sum_{k=1}^K\sqrt{\theta_{k}[n]}}, \forall n \in \mathcal{N}.
\end{eqnarray}
\end{proposition}
\begin{remark}\label{eta remark}
Note that the normalizing factor $\eta^{\star}[n]$  monotonically
increases with the  noise power $\sigma^2$. 
This indicates that as the noise power increases, a larger normalizing factor $\eta^{\star}[n]$ is required to suppress the noise-induced error, otherwise the time-averaged MSE will increase.
\end{remark}

%%%%%%%%%%%%%%%%%%%%%%%%%%%%%%%%%%%%%%%%%%%%

\subsection{Signal Quality Factors  Optimization}
In this subsection, we present the solution to problem $\mathscr{P}$ by optimizing  $\bm  \theta$  when $\bm q$ and $\bm \eta$ are fixed. The corresponding optimization problem is given by
\begin{subequations}
\begin{eqnarray}
\mathscr{P}_{1.2}:\mathop{\text{minimize}}_{\bm \theta} &&\sum_{n=1}^{N}\sum \limits_{k=1}^K\Bigg(\frac{\sqrt{\theta_{k}[n]}}{\eta[n]}-1\Bigg)^2\nonumber \\
	\text{subject to} 
&& \text{constraints} \ \eqref{thetacons: peakP}, \eqref{thetacons: averageP}, \nonumber
\end{eqnarray}
\end{subequations}
where the constant term $\{\sigma^2/\eta^2[n]\}$  is ignored in the objective function. 
In this case, we decompose problem $\mathscr{P}_{1.2}$ into the following $K$ subproblems  for optimizing $\theta_k[n]$, $\forall n \in \mathcal{N}$   to minimize the MSE for one sensor,
\vspace{-3mm}
\begin{subequations}\label{k-th theta}
\begin{eqnarray}
\mathop{\text{minimize}}_{\{ \theta_k[n]\}}&&\sum_{n=1}^{N}\Bigg(\frac{\sqrt{\theta_{k}[n]}}{\eta[n]}-1\Bigg)^2\nonumber \\
\text{subject to} 
&& 0\leq \frac{\theta_{k}[n]}{|h_k[n]|^2}\leq P_k,    \forall n, \label{k-thetacons: peakP}  \\ 
&& 0\leq\frac{1}{N}\sum_{n=1}^N \frac{\theta_{k}[n]}{|h_k[n]|^2}\leq \bar{P}_k.   \label{k-thetacons: averageP}
 \end{eqnarray}
\end{subequations}
 Note that problem \eqref{k-th theta} is a convex linearly constrained quadratic program (QP) that can be directly solved by using modeling framework CVX and interior-point solvers (e.g., SDPT3)  \cite{Grant2014CVX, Shi2015Largescale}, similar to  \cite{Shi2014GroupSparse, Chen2008CrossLayer}.
However, by exploring its special property, we can obtain more efficient solutions. 
Because strong duality holds between problem \eqref{k-th theta} and its Lagrange dual problem. We can leverage the Lagrange-duality method to optimally solve problem \eqref{k-th theta}.
Let $\alpha_n \geq 0$ denote the dual variable associated with the $n$-th constraint in \eqref{k-thetacons: peakP}, $\forall n\in \mathcal{N}$. 
Let $\lambda\geq 0$ denote the dual variable associated with the constraint \eqref{k-thetacons: averageP}.
Then the Lagrangian of problem \eqref{k-th theta} is
\begin{eqnarray}
 &&\hspace{-3em}\mathcal{L}(\!\{\theta_k[n]\}, \{\alpha_n\}, \lambda) \!=\!\!\sum_{n=1}^{N}\!\!\left(\frac{\sqrt{\theta_{k}[n]}}{\eta[n]}\!-\!1\right)^2 \!\!\!\!
 +\! \sum_{n=1}^{N}\! \alpha_n\!\left(\frac{\theta_{k}[n]}{|h_k[n]|^2} \!-\! P_k \!\right) \!+\! \lambda\!\left(\! \sum_{n=1}^N \frac{\theta_{k}[n]}{|h_k[n]|^2} \!-\! N\bar{P}_k\!\right)\!. 
\end{eqnarray}
By applying the Karush-Kuhn-Tucker (KKT) conditions, we obtain the following result.
\begin{proposition} \label{theta solution lemma}
With any given $\bm \eta$ and $\bm q$, the optimal solution  $\bm \theta$ of problem  $\mathscr{P}_{1.2}$ is given by
\begin{eqnarray} \label{theta solution}
&&\theta^{\star}_k[n] = 
\left\{
\begin{aligned}
& \min\big\{ \eta^2[n], P_k|h_{k}[n]|^2 \big\}, 
\ \text{if} \  \min\Big\{ \frac{\eta^2[n]}{|h_{k}[n]|^2}, P_k \Big\}\leq N\bar{P}_k,  \\
&\min \Big\{ \Big(\frac{\eta[n]|h_{k}[n]|^2}{\big|h_{k}[n]|^2 +  \lambda^{\star}\eta^2[n]}\Big)^2, P_k|h_{k}[n]|^2 \Big\}, \  \text{otherwise},
\end{aligned}
\right.
\end{eqnarray}
where  $\lambda^{\star}$ is a constant that ensures the average power constraint $\sum_{n=1}^N {\theta^{\star}_k[n]}/{|h_k[n]|^2} =  N\bar{P}_k$, which can efficiently be found  via a one-dimensional bisection search, though a closed-form expression is not attainable.
\end{proposition}	
\begin{proof}
	Please refer to Appendix \ref{theta solution lemma proof}.
\end{proof}
\begin{remark}\label{theta remark}
Note that if $\theta^{\star}_k[n] = \eta^2[n], \forall k, \forall n$, then the signal-misalignment error is zero, i.e., $\sum_{n=1}^N\sum_{k=1}^K\big({\sqrt{\theta_{k}[n]}}/{\eta[n]}-1\big)^2 = 0$.
Since the sensors' power budget is limited and  $\eta[n]$ is applied  to all sensors with different channels at each time slot,
$\theta^{\star}_k[n]$ may not be always equal  to $\eta^2[n]$ for all $n \in \mathcal{N}$, $k \in \mathcal{K}$.
If so, then $\theta^{\star}_k[n]$ has to be $\big({\eta[n]|h_{k}[n]|^2}/{(|h_{k}[n]|^2 +  \lambda^{\star}\eta^2[n])}\big)^2$ or $P_k|h_{k}[n]|^2$. 
%Note that $\big({\eta[n]|h_{k}[n]|^2}/{(|h_{k}[n]|^2 +  \lambda^{\star}\eta^2[n])}\big)^2\leq \eta^2[n]$. 
%And it is easily verified that when  $\theta_k[n] = P_k|h_{k}[n]|^2$,  $P_k|h_{k}[n]|^2 \leq \eta^2[n]$ must hold.
Observing that when   $\theta_{k}[n]\leq\eta^2[n]$,  the objective function of problem $\mathscr{P}_{1.2}$ %$\sum_{n=1}^{N}\sum_{k=1}^K\big({\sqrt{\theta_{k}[n]}}/{\eta[n]}-1\big)^2$
 monotonically decreases as $\theta_{k}[n]$ increases.
Thus, we can decrease the objective value by increasing the signal quality factors $\bm \theta$.
Meanwhile, both terms $\big({\eta[n]|h_{k}[n]|^2}/{(|h_{k}[n]|^2 +  \lambda^{\star}\eta^2[n])}\big)^2$ and 
$P_k|h_{k}[n]|^2$ in  \eqref{theta solution} monotonically increase with the increase of $|h_{k}[n]|^2$.
Based on the above analysis, it is found that we can further reduce the time-averaged MSE by increasing the sensors' channel power gains while still keeping  constraints \eqref{thetacons: peakP} and \eqref{thetacons: averageP} feasible.
\end{remark}
%%%%%%%%%%%%%%%%%%%%%%%%%%%%%%%%%%%%%%%%%%%%
\subsection{UAV Trajectory Optimization}
Next, we  optimize the UAV trajectory $\bm q$ for given $\bm \theta$ and $\bm \eta$.  Problem $\mathscr{P}$ is reduced to a feasibility checking problem, i.e.,
\begin{subequations} \label{q subproblem}
\begin{eqnarray}
\hspace{-1.5em} \mathop{\text{find}} && \bm q  \nonumber \\
\hspace{-1.5em}\text{subject to}
 && 0\leq \|\bm q[n] - {\bf w}_k[n]\|^2\leq \hat{P}_k[n],  \forall k,  \forall n, \label{qcons: peak p} \\
 && 0\leq\sum_{n=1}^N\theta_k[n]\|\bm q[n] - {\bf w}_k[n]\|^2\leq \tilde{P}_k,  \forall k, \label{qcons: average p}\\
 && \text{constriants} \ \eqref{qcons: speed}, \eqref{qcons: initial}, 
\end{eqnarray}
\end{subequations}
where $\hat{P}_k[n] = {\beta_0P_k}/{\theta_k[n]} - H^2$ and $\tilde{P}_k = N\beta_0\bar{P}_k - H^2\sum_{n=1}^{N}\theta_k[n]$.
% ------------------------------------------------------------
It is not difficult to observe that problem $\eqref{q subproblem}$ is a convex QCQP feasibility detection problem. 
To obtain a  more efficient solution and help reduce the value of the objective function at each iteration, 
we further transform problem $\eqref{q subproblem}$ into an optimization problem with an explicit objective.
Intuitively,  if the optimized trajectory by solving  \eqref{q subproblem} achieves a strictly smaller value of weighted-sum distance between the UAV and sensor $k$ than the corresponding threshold $\tilde{P}_k$, the channel gain between the UAV and sensor $k$ can be further increased since the channel gain is monotonically increased as the  distance between the UAV and the sensor decreases.
In turn, as presented in \textbf{Remark \ref{theta remark}}  of Section III, as channel gains between the UAV and the sensors increase, the signal quality factors $\{\theta_k[n]\}$ in problem $\mathscr{P}_{1.2}$ can be increased while satisfying all the power constraints, thereby reducing the time-averaged MSE.
To this end, problem $\eqref{q subproblem}$ is transformed into the following weighted-sum distance (between the UAV and all sensors) minimization problem
\begin{subequations} 
	\begin{eqnarray}
\mathscr{P}_{1.3}:	\mathop{\text{minimize}}_{\bm q} &&\sum_{n=1}^N\sum_{k=1}^{K}\theta_k[n]\|\bm q[n] - {\bf w}_k[n]\|^2 \nonumber \\
		\text{subject to} 
		&& \text{constriants} \ \eqref{qcons: speed}, \eqref{qcons: initial},  \eqref{qcons: peak p}, \eqref{qcons: average p},\nonumber
	\end{eqnarray}
\end{subequations}
where we set the weights to be $\theta_k[n],\forall k, \forall n$. 
This is because, in the AirComp setup, the distances between the UAV and different sensors cannot be simultaneously minimized in general at any time slot, which thus need to be balanced with different weights.
Comparing problem $\mathscr{P}_{1.3}$ with problem $\eqref{q subproblem}$, it follows that the feasible set of problem $\mathscr{P}_{1.3}$ is the same as problem $\eqref{q subproblem}$. 
However,  $\mathscr{P}_{1.3}$ is more practically desired than problem $\eqref{q subproblem}$ in terms of the converged solution as the time-averaged MSE decreases more quickly with the number of iterations.
Although the convex QCQP problem  $\mathscr{P}_{1.3}$ can be solved using a general-purpose solver through interior-point methods,  to further reduce the computational complexity, we exploit the specific structure of problem $\mathscr{P}_{1.3}$ and find its optimal solution using an ADMM-based algorithm in the sequel.

To utilize the ADMM method to solve problem $\mathscr{P}_{1.3}$, we introduce some necessary auxiliary variables and transform the constraints such that problem $\mathscr{P}_{1.3}$ can be decoupled into several convex subproblems as follows.
To begin with, we define $\mathbf A_1$ as
\begin{eqnarray}
\mathbf A_1 =\begin{bmatrix}
-1&  1&  & 0 & \cdots  &  0  & 0\\
0 & -1& 1& 0 &\cdots   &  0  & 0\\
0 & 0 & 0 & 0 & \ddots &  1  & 0\\
0 & 0 & 0 & 0 & \cdots & -1  & 1
\end{bmatrix}\in \mathbb R^{N\times (N+1)}.\nonumber
\end{eqnarray}
By introducing an auxiliary variable $\bm z = [\bm z[1]^{\sf T}, \ldots, \bm z[N]^{\sf T}]^{\sf T}\in \mathbb{R}^{N\times 2}$ with $\bm z[n] = \bm q [n]-\bm q[n-1]\in \mathbb{R}^{1\times 2}, \forall n \in \mathcal{N}$, constraint \eqref{qcons: speed} is equivalently expressed as
\begin{eqnarray}
\mathbf A_1\mathbf q=\bm z, \bm z \in \mathcal{Z},
\end{eqnarray}
  where $\mathcal{Z} = \left\{\bm z\in \mathbb{R}^{N\times 2} \ \big| \ \|\bm z[n]\| \leq  V_{\max}\delta, \forall n \in \mathcal{N} \right \}$ represents the feasible set of $\bm z$.
% ------------------------------------------------------------
Besides, constraint \eqref{qcons: initial}  can be equivalently expressed as
 \begin{eqnarray}
\mathbf A_2 \bm q = {\bf q}_I, 
\end{eqnarray}
where $\mathbf A_2 = [1, 0, 0, 0, \cdots, 0, 0] \in \mathbb R^{1\times (N+1)}$.
% ------------------------------------------------------------
Similarly, by denoting
\begin{eqnarray}
\hspace{-3em} &&\mathbf{B}_{1,k}  =   \text{diag}\left(1,\sqrt{\theta_k[1]}, \ldots, \sqrt{\theta_k[N]}\right)\in\mathbb{R}^{(N+1)\times (N+1)}, \\
\hspace{-3em}&&\mathbf{B}_{2,k}  =  [\bm q_I^{\sf T}, \mathbf{w}_k^{\sf T}]^{\sf T} \in \mathbb{R}^{(N+1)\times 2},
 \end{eqnarray}
constraint \eqref{qcons: peak p} and constraint \eqref{qcons: average p} can respectively be equivalently expressed as
\begin{eqnarray}
&& \| \bm q[n] - \mathbf{B}_{2,k}[n] \|_2^2 \leq \hat{P}_k[n], \forall k \in \mathcal{K}, \forall n \in \mathcal{N}, \\
 &&  \|\mathbf{B}_{1,k} (\bm q - \mathbf{B}_{2,k}) \|_F^2 \leq \tilde{P}_k,  \forall k,
\end{eqnarray}
where $\mathbf{B}_{2,k}[n]\in \mathbb{R}^{1\times 2}$ denotes the $(n+1)$-th row of matrix $\mathbf{B}_{2,k}$.
% ------------------------------------------------------------
To proceed, two sets of auxiliary variables $\{\bm \Gamma_k \in \mathbb{R}^{ (N+1)\times 2}, \forall k \in \mathcal{K}\}$  and $\{\bm V_{k}\in \mathbb{R}^{ (N+1)\times 2},\forall k \in \mathcal{K} \}$ are introduced
such that
\begin{eqnarray} 
    &&\bm \Gamma_k = \bm q, \forall k \in \mathcal{K}, \\
	&&\bm V_{k}      =  \bm B_{1,k}\bm q,  \forall k \in \mathcal{K}, 
\end{eqnarray}
where $\bm \Gamma_k $ is a copy of the original trajectory vector $\bm q$, and $\bm V_{k}$ represents the weighted trajectory vector.
To ease the notation, we define $\bm \Gamma \triangleq \big\{\bm \Gamma_{k}  \big|  \forall k \in \mathcal{K}  \big\}$, and  $\bm V \triangleq \big\{\bm V_{k}  \big|  \forall k \in \mathcal{K}  \big\}$.
% ------------------------------------------------------------
Then problem $\mathscr{P}_{1.3}$ can be equivalently expressed as
\begin{subequations} \label{Problem q24}
\begin{eqnarray}
\hspace{-3em}\mathop{\text{minimize}}_{ \bm q,  \bm \Gamma, \bm V , \bm z}  && \sum_{k=1}^K\| \mathbf{B}_{1,k}\bm q - \mathbf{B}_{1,k}\mathbf{B}_{2, k} \|_F^2 \nonumber\\
\hspace{-3mm} \text{subject to} 
&&\bm \Gamma_{k}  = \bm q , \forall k \in \mathcal{K}, \label{q24006}\\
&&\bm V_{k}   =  \mathbf{B}_{1,k}\bm q,  \forall k \in \mathcal{K}, \label{q24007}\\
&& 0\leq \|\bm \Gamma_k[n] - \mathbf{B}_{2,k}[n]\|^2\leq \hat{P}_k[n],  \forall k \in \mathcal{K}, \forall n\in \mathcal{N}, \label{q24001} \\
&&\| \bm V_{k} - \mathbf{B}_{1,k}\mathbf{B}_{2, k} \|_F^2 \leq \tilde{P}_k,  \forall k \in \mathcal{K}, \label{q24002}\\
&&\mathbf{A}_1 \bm q = \bm z,  \label{q24004}\\
&& \bm z \in \mathcal{Z},\\
&&\mathbf{A}_2 \bm q = \mathbf{q}_I, \label{q24005}
\end{eqnarray}
\end{subequations}
% ------------------------------------------------------------
where $\bm \Gamma_k[n]\in \mathbb{R}^{1\times 2}$ denotes the $(n+1)$-th row of matrix $\bm \Gamma_k$. 

We define the feasible regions of constraints \eqref{q24001},  \eqref{q24002}, and \eqref{q24005} as $\mathcal{C}$, $\mathcal{D}$, and $\mathcal{Q}$, respectively.
Thus, we obtain the equivalent ADMM reformulation of problem $\mathscr{P} _{1.3}$ as
\begin{eqnarray} \label{Problem ADMM q25}
\mathop{\text{minimize}}_{\bm \Gamma, \bm V, \bm q, \bm z} \hspace{-0mm}
&& \sum_{k=1}^K\| \mathbf{B}_{1,k}\bm q - \mathbf{B}_{1,k}\mathbf{B}_{2, k} \|_F^2
 + \ \mathbb{I}_{\mathcal{C}}(\bm \Gamma) +  \mathbb{I}_{\mathcal{D}}(\bm V)  + \mathbb{I}_{\mathcal{Z}}(\bm z)+ \mathbb{I}_{\mathcal{Q}}(\bm q)\nonumber  \\
\text{subject to} 
&& \text{constriants} \ \eqref{q24006}, \eqref{q24007}, \eqref{q24004},
\end{eqnarray}
where $\mathbb{I}_{\mathcal{X}}(\bm x)$ is  the indicator function for the feasible region of $\mathcal{X}$, which is given by
\begin{eqnarray}
\mathbb{I}_{\mathcal{X}}(\bm x) =\left\{
\begin{aligned}
 0,         &&\text{if} \ \bm x \in \mathcal{X},  \\
 +\infty,   &&\text{otherwise}.
\end{aligned}
\right.
\end{eqnarray}
Then, the augmented Lagrangian (using the scaled dual variables) of problem \eqref{Problem ADMM q25} is given by
\begin{eqnarray} \label{Lagrangian function}
 \mathcal{L}_\rho(\bm \Gamma, \bm V, \bm z, \bm q, \bm \lambda, \bm \xi, \bm \tau) 
& \!\!\!\!\!\!=\!\!\!\! \!\!\!&\sum_{k=1}^K\| \mathbf{B}_{1,k}\bm q - \mathbf{B}_{1,k}\mathbf{B}_{2, k} \|_F^2
 + \ \mathbb{I}_{\mathcal{C}}(\bm \Gamma) +  \mathbb{I}_{\mathcal{D}}(\bm V)+ \mathbb{I}_{\mathcal{Z}}(\bm z) + \mathbb{I}_{\mathcal{Q}}(\bm q) \nonumber \\
&&{\!\!\!\! \!\!\!\!\!\!\! \!\!\!\!\!\!\! \!\!\!\!\!\!\! \!\!\!\!\!\!\! \!\!\!\!\!\!\! \!\!\!}+  \frac{\rho_1}{2} \sum_{k = 1}^K \| \bm\Gamma_k - \bm q +   \bm \lambda_k\|_F^2  
  +  \frac{\rho_2}{2} \sum_{k=1}^{K} \| \bm V_{k} - \bm B_{1,k}\bm q  + \bm \xi_{k}\|_F^2 
 + \frac{\rho_3}{2} \|\bm z -\mathbf{A}_1 \bm q + \bm \tau\|_F^2,
\end{eqnarray}
where $\rho_1, \rho_2$, and $\rho_3$ are the penalty parameters,  and $\bm \lambda \triangleq \big\{\bm \lambda_{k} \in \mathbb{R}^{(N+1)\times 2}\big|  k \in \mathcal{K} \big\}$, $\bm \xi \triangleq \big\{\bm \xi_{k} \in \mathbb{R}^{(N+1)\times 2}\big|  k \in \mathcal{K}  \big\}$, and $\bm \tau \in \mathbb{R}^{N \times 2}$ are the dual variables for constraints \eqref{q24006}, \eqref{q24007}, \eqref{q24004}, respectively.
% ------------------------------------------------------------
\begin{algorithm}[t]
	\caption{ADMM for Solving Problem  $\mathscr{P}_{1.3}$}
	\begin{algorithmic}[1]
		\STATE {\textbf{Input}}:  The penalty parameters $\{\rho_1, \rho_2,\rho_3\}$.
		\STATE {{Initialize}}:  $\bm q^0 \leftarrow \bm q^i$, $\bm \lambda^0\leftarrow \bm 0$, $\bm \xi^0 \leftarrow \bm 0$, $\bm \tau^0\leftarrow \bm 0$.  Let $j = 0$. 
		\REPEAT		
		\STATE	Update the first block of variables $\{\bm \Gamma, \bm V, \bm z\}$	
		\begin{equation*} \label{first block}	
		\{\bm \Gamma^{j+1}, \bm V^{j+1}, \bm z^{j+1}\}:= \argmin_{\bm \Gamma, \bm V, \bm z}\mathcal{L}_\rho(\bm \Gamma, \bm V, \bm z, \bm q^j, \bm \lambda^j, \bm \xi^j, \!\bm \tau^j). 
		\end{equation*} 
		\STATE Update the second block of variables $\bm q$	
		\begin{equation*}\label{second block}	
			\bm q^{j+1}:= \argmin_{\bm q}\mathcal{L}_\rho(\bm \Gamma^{j+1}, \bm V^{j+1}, \bm z^{j+1}, \bm q, \bm \lambda^j, \bm \xi^j, \bm \tau^j ). 
		\end{equation*} 
		\STATE Update the dual variables based on expressions  \eqref{dual 1}, \eqref{dual 2}, and \eqref{dual 3}.
		\STATE Set $j = j+1$.
		\UNTIL convergene criterion is met.
		\STATE {\textbf{Output}}:  $\{\bm \Gamma, \bm V, \bm z, \bm q, \bm \lambda, \bm \xi, \bm \tau\}$. 
	\end{algorithmic}
	\label{algo2.3}
\end{algorithm}  
% ------------------------------------------------------------

According to  \eqref{Lagrangian function}, we observe that the primal variables can be split into two blocks, i.e., $\{\bm \Gamma, \bm V, \bm z  \}$ and $\bm q$, and the objective function is  also separable along with this splitting. 
Therefore, by adopting the ADMM, we can minimize $\mathcal{L}_\rho(\bm \Gamma, \bm V, \bm z, \bm q, \bm \lambda, \bm \xi, \bm \tau, \bm \eta )$ by alternately updating the two blocks of primal variables.
Note that the first block of variables $\{\bm \Gamma, \bm V, \bm z  \}$ can be decomposed into three independent problems, which are expressed as follows.
%Therein, the subproblems for updating  $\bm \Gamma$, $\bm V$, and $\bm z$ are expressed as follows, 
\begin{eqnarray} 
&&\bm \Gamma^{j+1}:= \argmin_{\bm \Gamma} \bigg\{  \mathbb{I}_{\mathcal{C}}(\bm \Gamma) + \frac{\rho_1}{2} \sum_{k=1}	\| \bm\Gamma_k - \bm q^j +   \bm \lambda^j\|_F^2\bigg\}, \label{ADMM Gamma}\\
&& \bm V^{j+1} := \argmin_{\bm V} \bigg\{  \mathbb{I}_{\mathcal{D}}(\bm V) + \frac{\rho_2}{2} \sum_{k=1}^{K} \big\| \bm V_{k} - \mathbf{B}_{1,k}\bm q^j + \bm \xi_{k}^j \big\|_F^2 \bigg\}, \label{ADMM V} \\ 
&&\bm z^{j+1} :=  \argmin_{\bm z} \bigg\{ \mathbb{I}_{\mathcal{Z}}(\bm z) 
+\frac{\rho_3}{2} \big\| \bm z - \mathbf{A}_1 \bm q^j + \bm \tau^j \big\|_F^2 \bigg\}. \label{ADMM z}
\end{eqnarray}
%%%%%%%%%%%%%%%%%%%%%%%%%%%%%%%%%%%%%%%%%%%%
In the $j$-th iteration, given $\{ \bm q^j, \bm \lambda^j, \bm \xi^j, \bm \tau^j\}$,  the details of updating each variable are explained as follows.
\subsubsection{$\bm \Gamma$ Update}  Problem \eqref{ADMM Gamma} for updating  $\bm \Gamma$  is equivalent to 
\begin{eqnarray}\label{Gamma Update}
	\mathop{\text{minimize}}_{ \bm\Gamma} 
 && \sum_{k=1}^K\| \bm\Gamma_k - \bm q^j +   \bm \lambda_k^j\|_F^2\nonumber \\
\text{subject to} 
&& 0\leq \|\bm \Gamma_k[n] - \mathbf{B}_{2,k}[n] \|^2\leq \hat{P}_k[n],  \forall k,  \forall n. 
\end{eqnarray}
One can  observe that problem \eqref{Gamma Update} can be decomposed into $KN$ convex subproblems, each of which is a QCQP with only one constraint (QCQP-1), and thus  is efficiently solvable. 
Specifically, one for each $n \in \mathcal{N}$, $k \in \mathcal{K}$ is given by
\begin{eqnarray}\label{kn-th Gamma Update}
	\mathop{\text{minimize}}_{ \bm\Gamma_k[n]} 
&&  \| \bm\Gamma_k[n] - \bm q^j[n] + \bm \lambda_k^j[n]\|^2\nonumber \\
\text{subject to} 
&& 0\leq \|\bm \Gamma_k[n] - \mathbf{B}_{2,k}[n] \|^2\leq \hat{P}_k[n],
\end{eqnarray}
where $\bm \lambda_k^j[n]\in \mathbb{R}^{1\times 2}$ denotes the $(n+1)$-th row of matrix $\bm \lambda_k^j$. 
 Problem
\eqref{kn-th Gamma Update} can be viewed as the Euclidean projection of the point
$\bm q^j[n] - \bm \lambda_k^j[n]$
onto an Euclidean ball, centered at the point $\mathbf{B}_{2,k}[n]$  with
radius of $\sqrt{\hat{P}_k[n]}$.
It is easily verified that the optimal solution of  subproblem \eqref{kn-th Gamma Update} is given as the following closed form
\begin{eqnarray}
\bm \Gamma_k[n] =\left\{
\begin{aligned}
 & \mathcal{P}_{\mathcal{C}} \big(\bm q^j[n] -\bm \lambda^j[n] - \mathbf{B}_{2,k}[n]\big) + \mathbf{B}_{2,k}[n],  k \in \mathcal{K},  n \in \mathcal{N},  \\
 &\bm q^{j}[n],  n = 0,
\end{aligned}
\right.
\end{eqnarray}
where $\mathcal{P}_{\mathcal{C}}(\bm x_k[n]) := \min\big\{{\sqrt{\hat{P}_k[n]}}\big/{\|\bm x_k[n]\|},1 \big\}\bm x_k[n] $ denotes the projector associated with the space $\mathcal{C}$.
%%%%%%%%%%%%%%%%%%%%%%%%%%%%%%%%%%%%%%%%%%%%

\subsubsection{$\bm V$ Update} Problem \eqref{ADMM V} for updating  $\bm V$ is equivalent to the following problem
\begin{eqnarray}\label{V Update}
	\mathop{\text{minimize}}_{ \bm V} 
  &&\sum_{k=1}^{K} \big\| \bm V_{k} - \mathbf{B}_{1,k}\bm q^j + \bm \xi_{k}^j \big\|_F^2 \nonumber \\
    \text{subject to} 
    &&  \|\bm V_{k} - \mathbf{B}_{1,k}\mathbf{B}_{2,k}\|_F^2\leq \tilde{P}_k,  \forall k.
\end{eqnarray}
It is observed that problem \eqref{V Update}  can be decomposed into $K$ QCQP-1 subproblems. One for each  $k \in \mathcal{K}$ is given by
\vspace{-2mm}
\begin{eqnarray}\label{k-v Update}
	\mathop{\text{minimize}}_{\bm V_{k}} 
	&& \big\| \bm V_{k} - \mathbf{B}_{1,k}\bm q^j +  \bm \xi_{k}^j \big\|_F^2 \nonumber \\
	\text{subject to} 
    && \|\bm V_{k} - \mathbf{B}_{1,k}\mathbf{B}_{2,k}\|_F^2\leq \tilde{P}_k. 
\end{eqnarray}
Similarly to problem \eqref{kn-th Gamma Update}, the optimal solution of problem \eqref{k-v Update} is given by
\begin{eqnarray}
	\bm V_{k} = \mathcal{P}_{\mathcal{D}} \big(  \mathbf{B}_{1,k}\bm q^j - \bm \xi_{k}^j - \mathbf{B}_{1,k}\mathbf{B}_{2,k}\big) + \mathbf{B}_{1,k}\mathbf{B}_{2,k},
\end{eqnarray}	
where $\mathcal{P}_{\mathcal{D}}(\bm X_k) := \min\big\{{\sqrt{\tilde{P_k}}}\big/{ \| \bm X_k\|_F },1 \big\}\bm X_k $ denotes the projector associated with  $\mathcal{D}$.
%%%%%%%%%%%%%%%%%%%%%%%%%%%%%%%%%%%%%%%%%%%%

\subsubsection{$\bm z$ Update} The update of $\bm z$ in problem \eqref{ADMM z} is equivalent to solving the following problem
\begin{eqnarray}\label{z Update}
	 \mathop{\text{minimize}}_{ \bm z} 
	 &&\big\| \bm z - \mathbf{A}_1 \bm q^j  + \bm \tau^j \big\|^2 \nonumber \\
	 \text{subject to} 
     &&  \|\bm z[n]\| \leq V_{\max}\delta,  \forall n \in \mathcal{N}.
\end{eqnarray}
Problem \eqref{z Update} can also be decomposed into $N$ QCQP-1 subproblems. One for each  $n \in \mathcal{ N}$ is given by
\vspace{-4mm}
\begin{eqnarray}\label{n-z Update}
	\mathop{\text{minimize}}_{ \bm z[n]} && \| \bm z[n] - \bm q^j[n] + \bm q^j[n-1]   + \bm \tau^j[n]\|^2 \nonumber \\
	\text{subject to} 
&&  \|\bm z[n]\| \leq V_{\max}\delta, 
\end{eqnarray}
where $\bm \tau[n]$ is the $n${-}th row of $\bm \tau$.
 
Similarly to problem \eqref{kn-th Gamma Update}, the optimal solution of problem \eqref{n-z Update} is obtained as
\begin{eqnarray}
	\bm z[n] = \mathcal{P}_{\mathcal{Z}}(\bm q^j[n] - \bm q^j[n-1]- \bm \tau^j[n]),
\end{eqnarray}
where $\mathcal{P}_{\mathcal{Z}}(\bm x) := \min\big\{{V_{\max}\delta}\big/{ \| \bm q^j[n] - \bm q^j[n-1]- \bm \tau^j[n]\|_2 },1 \big\}\bm x $ denotes the projector associated with the space $\mathcal{Z}$.
%%%%%%%%%%%%%%%%%%%%%%%%%%%%%%%%%%%%%%%%%%%%%%%%%%%%
\begin{algorithm}[t]
	\caption{Proposed  BCD-ADMM Algorithm for Solving Problem  $\mathscr{P}$}
	\begin{algorithmic}[1]
		\STATE {\textbf{Input}}:  $T$, $K$, $\{P_k\}$, $\{\bar{P}_k\}$, accuracy  $\epsilon$
		\STATE {{Initialize}}:  trajectory $ \bm q^0$  and  $\bm \theta^0$.  Let $i = 0$. 
		\REPEAT
		\STATE Set $i = i+1$.
		\STATE	Given  $ \bm q^{i-1}$ and $\bm \theta^{i-1}$, solve  $\mathscr{P}_{1.1}$  to  update   $\bm \eta^{i}$  based on  \eqref{solution denoise}.
		\STATE Given $ \bm q^{i-1}$ and $\bm \eta^{i}$, solve  $\mathscr{P}_{1.2}$  to update $\bm \theta^{i}$  based  on \eqref{theta solution}.
		\STATE Given  $\bm \eta^{i}$ and $\bm \theta^{i}$, solve  $\mathscr{P}_{1.3}$ to update $\bm q^{i}$  based on Algorithm \ref{algo2.3}.
	%	\STATE Calculate $R^i = \overline{\sf {MSE}}^i$.
		\UNTIL The relative decrease  $\frac{\overline{\sf {MSE}}^{i-1}-\overline{\sf {MSE}}^i}{\overline{\sf {MSE}}^i} < \epsilon$.
		\STATE {\textbf{Output}}:  $\bm \eta$, $\bm \theta$, and $\bm q$. 
	\end{algorithmic}
	\label{algo1}	
\end{algorithm}  

\subsubsection{$\bm q$ Update}
The update of  $\bm q$ is rewritten as the following linearly constrained QP problem
\begin{eqnarray}\label{q Update}
\mathop{\text{minimize}}_{ \bm q}&& 
	\sum_{k=1}^K\| \mathbf{B}_{1,k}\bm q - \mathbf{B}_{1,k}\mathbf{B}_{2, k} \|_F^2 
+\frac{\rho_1}{2}\sum_{k = 1}^K \| \bm\Gamma_k - \bm q +   \bm \lambda_k\|_F^2  \nonumber \\ 
 && +\frac{\rho_2}{2} \sum_{k=1}^{K} \| \bm V_{k} -\bm B_{1,k}\bm q  + \bm \xi_{k}\|_F^2 
 +\frac{\rho_3}{2} \|\bm z - \mathbf{A}_1 \bm q + \bm \tau\|_F^2 \nonumber \\
 \text{subject to} &&
      \mathbf{A}_2 \bm q = {\bf q}_I.
\end{eqnarray}
The solution to this least square problem with an affine constraint can be obtained based on the orthogonal projection onto an affine subspace, whose closed form is given by
\begin{eqnarray} 
&&\!\!\!\!\!\!\!\!\!\!\!\bm J =  \sum_{k=1}^K\Big[2\mathbf{B}^{\sf T}_{1,k}\mathbf{B}_{1,k}\mathbf{B}_{2,k} + \rho_1\big(\bm \Gamma_k^{j+1} + \bm \lambda_k^{j}\big) \Big] 
+ \sum_{k=1}^K\rho_2\mathbf{B}^{\sf T}_{1,k}\big(\bm V_k^{j+1} + \bm \xi_k^{j}\big) +\rho_3\mathbf{A}^{\sf T}_1\big(\bm z^{j+1} + \bm \tau^{j}\big),  \\
&&\!\!\!\!\!\!\!\!\!\!\!\bm F= {\rho_1K\bm I + \sum_{k=1}^K(\rho_2 + 2)\mathbf{B}^{\sf T}_{1,k}\mathbf{B}_{1,k} + \rho_3\mathbf{A}^{\sf T}_1\mathbf{A}_1}, \\
&&\!\!\!\!\!\!\!\!\!\!\!\bm q =  \big(\bm I - \mathbf{A}_2^{\sf T}\mathbf{A}_2\big)\big({\bm F}^{-1}{\bm J}\big)+ \mathbf{A}_2^{\sf T}{\bf q}_I. \label{q solution}     
\end{eqnarray}
%%%%%%%%%%%%%%%%%%%%%%%%%%%%%%%%%%%%%%%%%%%%
\subsubsection{Lagrange Multipliers Update}
The scaled dual variables are updated as 
\begin{eqnarray}
        &&\bm \lambda_{k}^{j+1} :=  \bm \lambda_{k}^{j} +\bm\Gamma_{k}^{j+1} -\bm q^{j+1}, \forall k\in\mathcal{K}, \label{dual 1}\\   
      &&\bm \xi_{k}^{j+1}       := \bm \xi_{k}^{j} + \bm V_{k}^{j+1} - \mathbf{B}_{1,k} \bm q^{j+1}, \forall k\in\mathcal{K}, \label{dual 2} \\ 
      && \bm \tau^{j+1}          := \bm \tau^{j} +  \bm z^{j+1}- \mathbf{A}_1 \bm q^{j+1}\label{dual 3}.
     \end{eqnarray}

%%%%%%%%%%%%%%%%%%%%%%%%%%%%%%%%%%%%%%%%%%%%

Up to now, the closed-form expressions for all the variables updating have been derived. 
To be specific, the proposed ADMM algorithm for solving problem $\mathscr{P}_{1.3}$ is summarized in Algorithm \ref{algo2.3}.
Note that Algorithm \ref{algo2.3} is guaranteed to converge to an  optimal solution of the convex problem $\mathscr {P}_{1.3}$ for any initial point with the convergence rate of $\mathcal{O}(1/j)$\cite{Bertsekas1997}.
For Algorithm \ref{algo2.3},  the dominant computation is the matrix inversion for updating $\bm q$ in expression \eqref{q solution}, whose complexity is $\mathcal{O}((N+1)^3)$. 
However, this operation only needs to be computed once during the iterative procedure of  Algorithm \ref{algo2.3}. 

\vspace{-2mm}
\subsection{Convergence and Complexity Analysis}
In the proposed BCD-ADMM method, we solve problem $\mathscr{P}$ by solving  $\mathscr{P}_{1.1}$, $\mathscr{P}_{1.2}$, and $\mathscr{P}_{1.3}$ alternately until convergence,  whose details are summarized in Algorithm $\ref{algo1}$.
Note that the solution $\bm q$ obtained in each iteration is used as the initial point of the next iteration for Algorithm $\ref{algo2.3}$. 
%We summarize the proposed BCD-ADMM algorithm in Algorithm $\ref{algo1}$
%in details.
The convergence of Algorithm $\ref{algo1}$ is proved by the following proposition.
\begin{proposition}\label{convergence proposition}
The objective value of Problem $\mathscr{P}$ decreases as the number of iteration increases until convergence by applying Algorithm $\ref{algo1}$.	
\end{proposition}
\begin{proof}
	Please refer to Appendix \ref{convergence proposition proof}.
\end{proof}

In the following, we investigate the complexity per iteration of  Algorithm \ref{algo1}.
Specifically, in step 5, the complexity for computing $\bm \eta$ is  $\mathcal{O}(N)$. 
In step 6, the complexity for computing $\bm \theta $ is $\mathcal{O}(KN)$. 
In step 7, the complexity for computing $\bm q$ mainly lies in computing expression \eqref{q solution}. 
Therein, the matrix inversion in \eqref{q solution} is computational intensive operation in Algorithm \ref{algo2.3}, with the complexity given by $\mathcal{O}((N+1)^3)$. 
It should be mentioned that the matrix inversion in \eqref{q solution} only needs to be computed once in Algorithm \ref{algo2.3}. 
In addition, the complexity for computing matrix multiplication  in \eqref{q solution}  is given by $\mathcal{O}((N+1)^2)$.
Since $N>K$  in this paper, the total complexity of our proposed algorithm is thus dominated by $\mathcal{O}((N+1)^2)$ in each iteration.
By comparison, the complexity in the conventional BCD-SCA method \cite{Fu2021UAV} for computing $\bm q$ is given by $\mathcal{O}(K^{1.5}N^{3.5})$ . In summary, the proposed BCD-ADMM algorithm significantly reduces the computational complexity in each iteration.

\begin{table*}[t]
	\centering
	\caption{Parameter settings for simulations.}\label{table:setupNumerical}
	{\renewcommand{\arraystretch}{1.2}\begin{tabular}{|p{3cm}|p{4cm}|p{4.5cm}|p{2.5cm}|}
			\hline
			Number of sensors & $K=50$  &	UAV altitude & $H=100$ m \\
			\hline
			Peak power budget & $P_A\!=\!10$ dBm, $P_B = 7$ dBm & Maximum UAV speed & ${V}_{\max}=20$ m/s			 \\
			\hline
			Average power budget & $\bar{P}_k = \frac{1}{2}P_k$ & Initial horizontal location of UAV & $\bm q_I = (200, 0) $ m
			\\
			\hline
			Noise power & $\sigma^2 = -80$ dBm & Channel gain at reference distance & $\beta_0=-40$ dB \\
			\hline
			Time slot length & $\delta = 0.2$ s &
			Accuracy of Algorithm \ref{algo1} &$\epsilon= 10^{-3}$
			\\
			\hline	 
			\end{tabular}}
		\vspace{-8mm}
\end{table*}
\section{Numerical Results}\label{NR section}
In this section, we present the numerical results to demonstrate the effectiveness of the proposed algorithm.
The service region of the UAV is limited to be a square area with the size of [0, 400] m $\times$  [0, 400] m. 
The UAV is assumed to fly at a fixed altitude of $H = 100$ m, which complies with the practical rule, i.e., the commercial UAVs  should not fly over 400 feet (122 m) \cite{FAA2016UAV}. 
Additionally, we consider a heterogeneous sensor network,
where $K=50$  sensors are separated into two clusters, i.e., cluster $A$ with 15 sensors and cluster $B$ with 35 sensors. 
The peak powers of sensors in  clusters $A$ and $B$ are denoted as $P_A$ and $P_B$, respectively. 
To investigate the performance of the BCD-ADMM algorithm, we consider a simplified mobility model to account for time-varying locations of the sensors, 
as in \cite{Zhang2021UAVmobility}.
In particular,  instead of considering a certain mobility model for individual sensors, we assign different traces for two clusters.
Following the trace, the center of each cluster continuously changes within the service region during mission duration $T$.
Specifically, the initial locations of the sensors in cluster $A$ and cluster $B$ are randomly and uniformly distributed in a circle centered at (50, 100) m and (350, 150) m with a radius of 50 m, respectively.
The cluster centers move at random directions and constant speeds that follow uniform distributions within the intervals  [0, $\pi$] and [1, 8] m/s, respectively.
Note that the proposed approach can be applied to other mobility models as long as the speeds of the sensors are lower than that of the UAV such that  the sensors' locations can be considered invariant within one time slot.

We compare the proposed BCD-ADMM algorithm with the following benchmarks.
\begin{itemize}
\item \textbf{BCD-SCA}: The BCD-SCA algorithm \cite{Fu2021UAV} jointly optimizes $\bm \eta$, $\bm p$, and $\bm q$ to solve Problem \eqref{original formulation}, wherein SCA  is adopted to approximate the non-convex trajectory design subproblem as a convex QCQP problem that can be efficiently solved by using modeling framework CVX and interior-point solvers (e.g., SDPT3).
\item Trajectory optimization without transmit power control (\textbf{TO w/o PC}): As considered in \cite{Fu2021UAV}, the TO w/o PC algorithm optimizes $\bm \eta$ and $\bm q$ with  constant transmit powers $\bm p$, i.e., $p_k[n] = \bar{P}_k, \forall k, \forall n$.
\item \textbf{Static UAV}:  In this scheme, the UAV is placed at the predetermined initial position (200, 0, 100) m and remains static. 
This scheme optimizes $\bm \eta$ and $\bm \theta$ by solving problems $\mathscr{P}_{1.1}$ and $\mathscr{P}_{1.2}$ iteratively until convergence.
\item Fly-hover with power control (\textbf{Fly-hover w/ PC}): The Fly-hover w/ PC algorithm designs the UAV trajectory in the following best-effort manner. 
If time allows, the UAV  flies straightly at its maximum speed to reach over the point of the sensors' geometric center of the last time slot and then remains static. 
Otherwise, it will fly straightly  at its maximum speed to reach over a point which is on the line between the UAV initial location and the sensors' geometric center of the last time slot.
Given the trajectory, the Fly-hover w/ PC algorithm optimizes $\bm \eta$ and $\bm \theta$ by solving problems $\mathscr{P}_{1.1}$ and $\mathscr{P}_{1.2}$ iteratively until convergence.
\end{itemize}
The initial transmit power  is $p_k[n] = \bar{P}_k, \forall k, \forall n$.
The initial trajectory for BCD-ADMM, BCD-SCA, and TO w/o PC are generated by using the Fly-hover w/ PC algorithm, unless specified otherwise.
For the BCD-based algorithms, the iterative procedure stops when the relative decrease of the objective is smaller than $10^{-3}$ (i.e., $(\overline{\sf{MSE}}^{i-1} - \overline{\sf{MSE}}^{i})/{\overline{\sf{MSE}}^i} \leq 10^{-3}$) or a maximum of 100 iterations is reached, where $\overline{\sf{MSE}}^i$ denotes the objective value of the $i$-th iteration. 
The penalty parameters in Algorithm \ref{algo2.3} are set as $\rho_1 = {8}/{\sqrt{K}}$, $\rho_2 = {8}/{\sqrt{K}}$, and $\rho_3 = {20}/{\sqrt{K}}$, which are empirically found to work well. 
The convergence criteria of ADMM are set as the absolute tolerance $\epsilon^{\rm abs} = 10^{-4}$ and relative tolerance $\epsilon^{\rm rel} = 10^{-4}$. 
Other parameters are summarized in Table \ref{table:setupNumerical} (if not specified otherwise).

\subsection{Convergence Behavior and Complexity Comparison}
\begin{figure}[t]
	\centering
	\begin{minipage}{.48\textwidth}
		\centering
		\includegraphics[width=8cm,height=6cm]{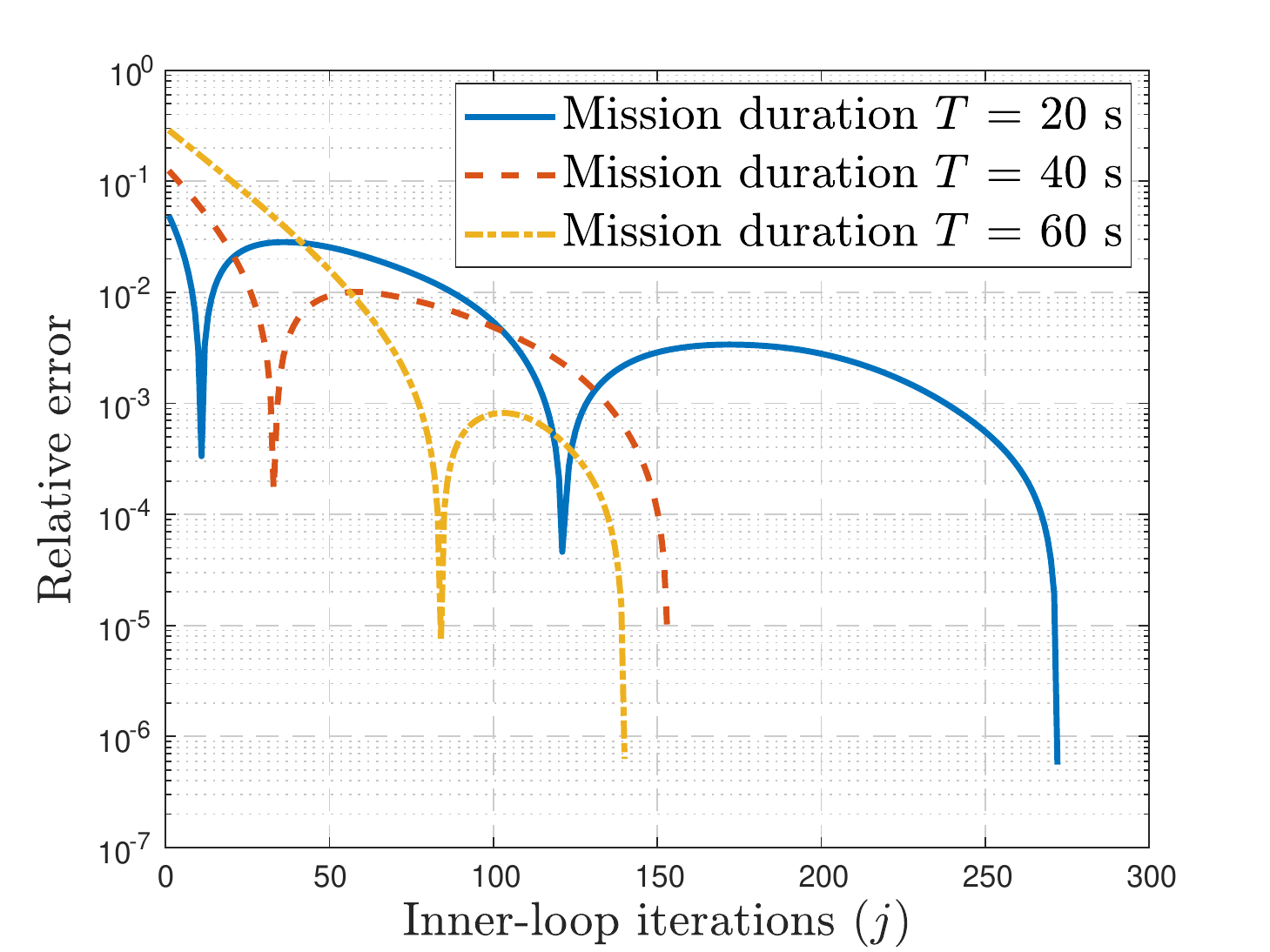}
		\vspace{-10mm}
		\caption{Convergence behavior of ADMM algorithm in the inner loop.}\label{fig:inner}
		\vspace{-8mm}
	\end{minipage}
	\hspace{4mm}
	\begin{minipage}{.48\textwidth}
		\centering
	\includegraphics[width=8cm,height=6cm]{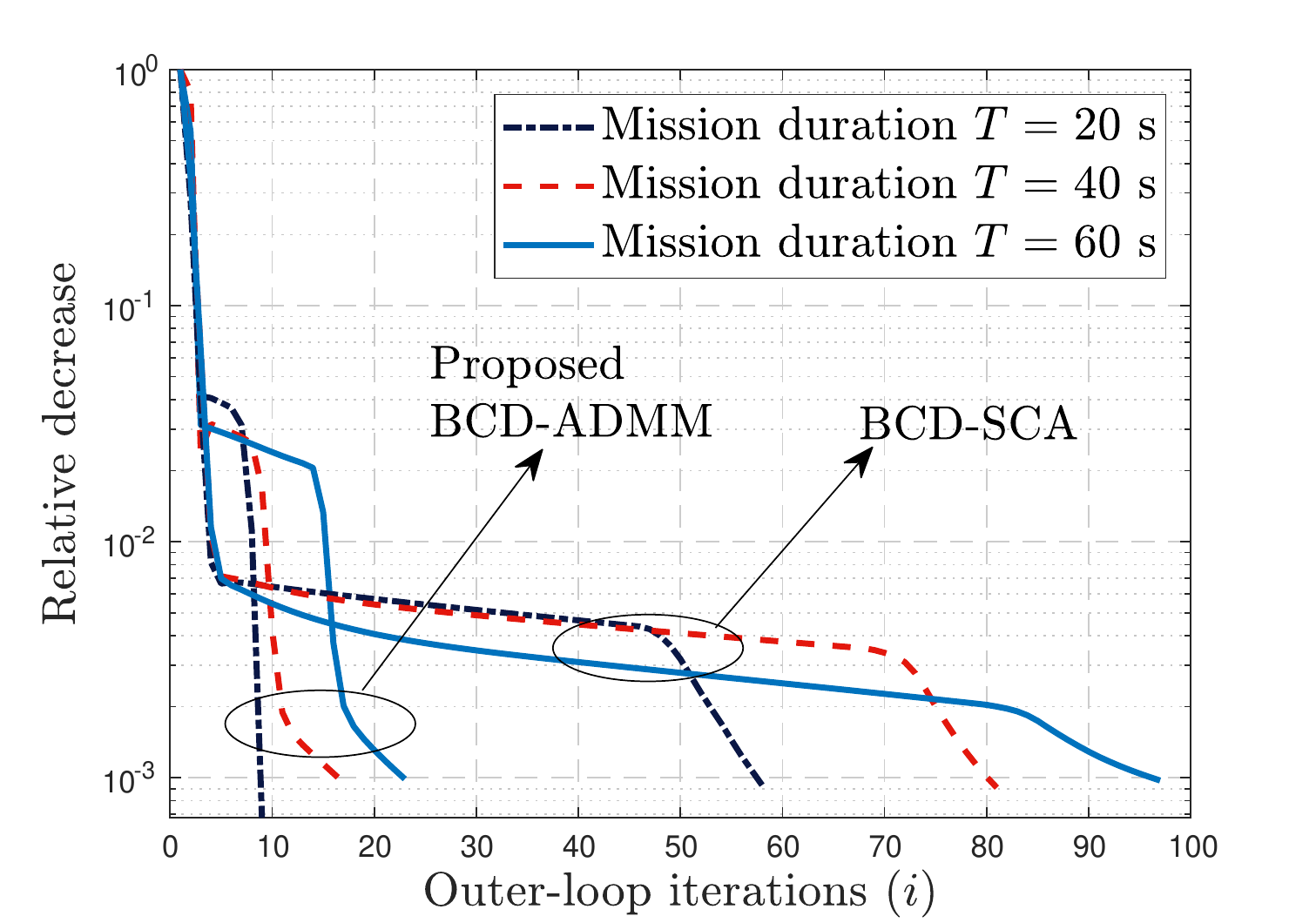}
		\vspace{-10mm}
			\caption{Convergence behavior of BCD algorithm in the outer loop.}\label{fig:outer}
		\vspace{-8mm}
	\end{minipage}
\end{figure}

%\begin{figure}[t]	
%		\centering
%		\includegraphics[width=8cm,height=6cm]{Figs/inner}
%		\vspace{-1mm}
%		\caption{Convergence behavior of ADMM algorithm in the inner loop.}\label{fig:inner}
%		\vspace{-1mm}
%\end{figure}
%\begin{figure}[t]	
%		\centering		\includegraphics[width=8cm,height=6cm]{Figs/outer}
%		\vspace{-1mm}
%		\caption{Convergence behavior of BCD algorithm in the outer loop.}\label{fig:outer}
%		\vspace{-1mm}
%\end{figure}	
We first demonstrate the convergence behavior of the proposed BCD-ADMM algorithm. 
The algorithm involves  an inner-loop iteration for ADMM to solve Problem $\mathscr P_{1.3}$ and an outer-loop iteration for BCD to solve Problem $\mathscr P$, whose convergences are illustrated in Fig. \ref{fig:inner} and Fig. \ref{fig:outer}, respectively.

In Fig. \ref{fig:inner},  the relative error in each iteration is defined as ${|a^j - a^\star|}/{a^\star}$, where $a^j$ is the objective value of the $j$-th iteration in Algorithm 1 and $a^\star$ is the optimal objective value of  problem $\mathscr P_{1.3}$ by using the interior-point method. 
From Fig. \ref{fig:inner},  it is seen that the proposed ADMM algorithm converges to a high accuracy solution, e.g., $10^{-5}$, within 300 iterations for various values of $T$. 
Note that the ADMM  is  not guaranteed to converge monotonically, since $\bm q^j$ is not generated to be feasible at each iteration  \cite{Bertsekas1997}.
From Fig. \ref{fig:outer}, it is observed that the proposed BCD-ADMM algorithm converges to a modest accuracy, e.g., $10^{-3}$, within 25 iterations for different values of $T$, while the BCD-SCA method needs about $60 \sim 100$ iterations to achieve the same level of accuracy.
This is because the proposed BCD-ADMM algorithm obtains the optimal solution of each subproblem, while the BCD-SCA method only optimizes the approximate lower bound of the trajectory subproblem by the SCA framework.

\begin{figure}[t]
	\centering
	\begin{minipage}{.48\textwidth}
		\centering
			\includegraphics[width=8cm,height=6cm]{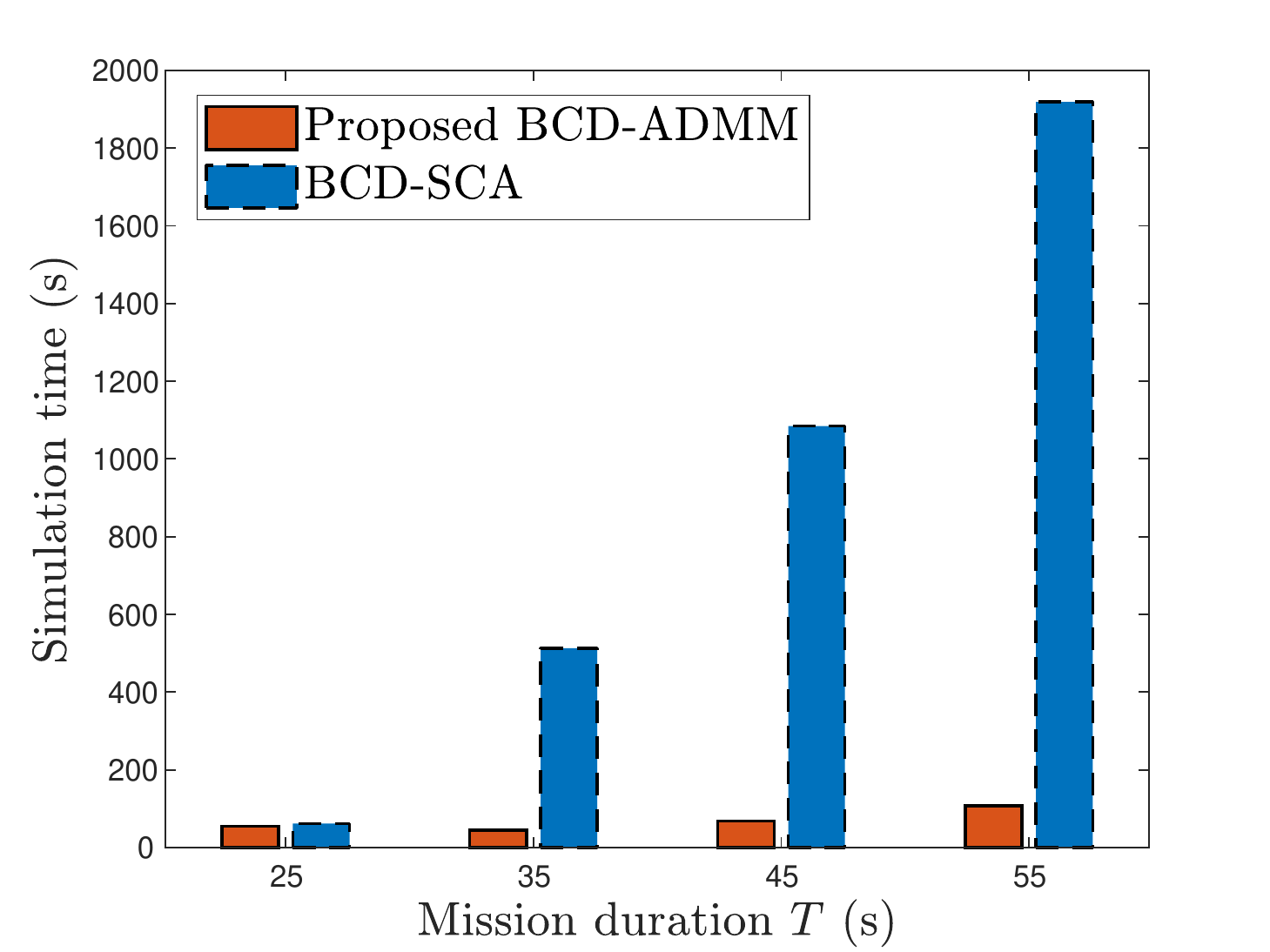}
		\vspace{-10mm}
		\caption{Simulation time versus the mission duration $T$.}\label{fig:timeT}
		\vspace{-8mm}
	\end{minipage}
	\hspace{4mm}
	\begin{minipage}{.48\textwidth}
		\centering
		\includegraphics[width=8cm,height=6cm]{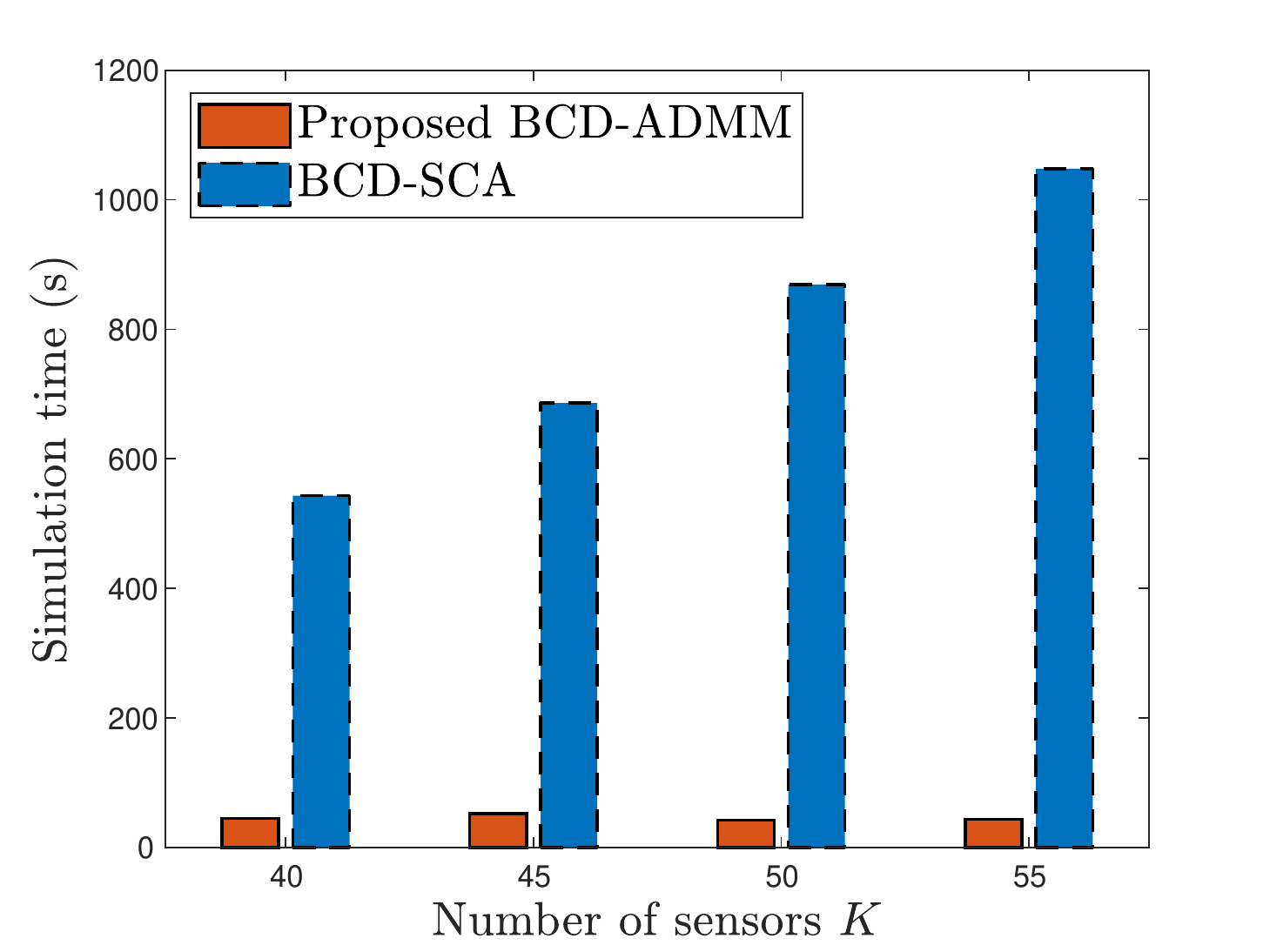}
		\vspace{-10mm}
			\caption{Simulation time versus the number of sensors $K$.}\label{fig:timeK}
		\vspace{-8mm}
	\end{minipage}
\end{figure}

%\begin{figure}[t]
%		\centering		\includegraphics[width=8cm,height=6cm]{Figs/TimeT}
%		\vspace{-1mm}
%		\caption{Simulation time versus the mission duration $T$.}\label{fig:timeT}
%		\vspace{-1mm}
%\end{figure}
%
%\begin{figure}[t]
%		\centering		\includegraphics[width=8cm,height=6cm]{Figs/TimeK}
%		\vspace{-1mm}
%		\caption{Simulation time versus the number of sensors $K$.}\label{fig:timeK}
%		\vspace{-1mm}	
%\end{figure}
Next, we compare the computational complexity of different methods in terms of the  simulation running time.
Fig. \ref{fig:timeT} shows  the simulation running times with different mission duration $T$. 
It is observed that as the mission duration $T$ increases, the running time of BCD-ADMM grows slowly, while that of BCD-SCA dramatically rises up.
Compared with the BCD-SCA method, the proposed BCD-ADMM method can speed up the running time about $10 \sim 35$ times.
This is because, for solving the trajectory optimization subproblem, the proposed ADMM only requires arithmetic operations rather than the interior-point solvers (e.g., SDPT3) in the BCD-SCA method.
Besides, the BCD-SCA method converges much slower than our proposed BCD-ADMM algorithm in the outer loop, as shown in Fig. \ref{fig:outer}.

Fig. \ref{fig:timeK}  illustrates the complexity comparison with different number of sensors $K$ when $T = 50$~s. Similar results as in the previous figure can be observed.
It is also seen that the running time of the proposed BCD-ADMM method is about $10 \sim 20$ shorter than the BCD-SCA method.

\subsection{Comparison of Different Trajectory Designs}

\begin{figure}[t]
	\centering
	\begin{minipage}{.48\textwidth}
		\centering
		\includegraphics[width=8cm,height=6cm]{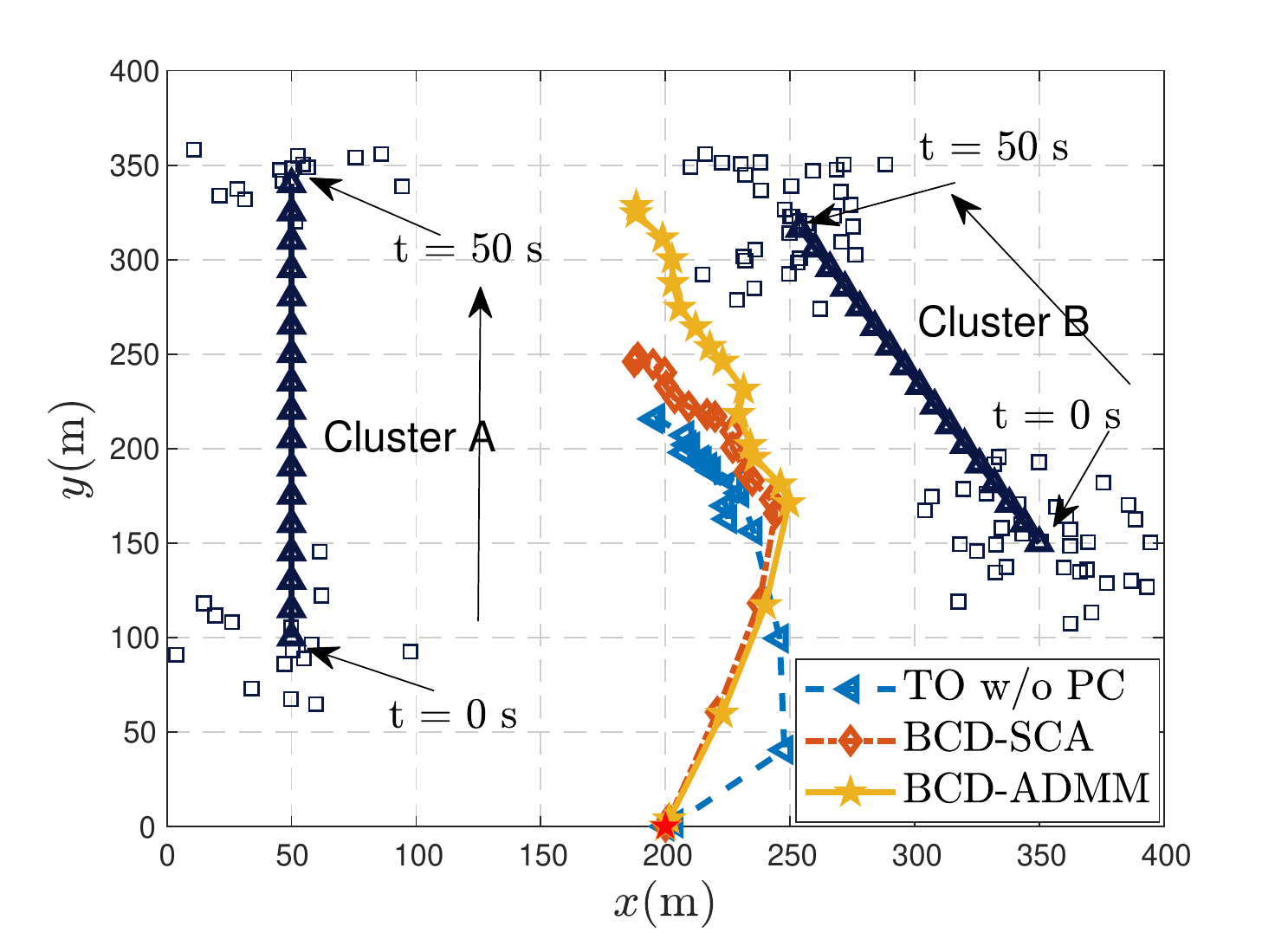}
		\vspace{-10mm}
			\caption{UAV trajectories for different algorithms when $T = 50$~s. }\label{fig:Traj}
		\vspace{-8mm}
	\end{minipage}
	\hspace{4mm}
	\begin{minipage}{.48\textwidth}
		\centering
			\includegraphics[width=8cm,height=6cm]{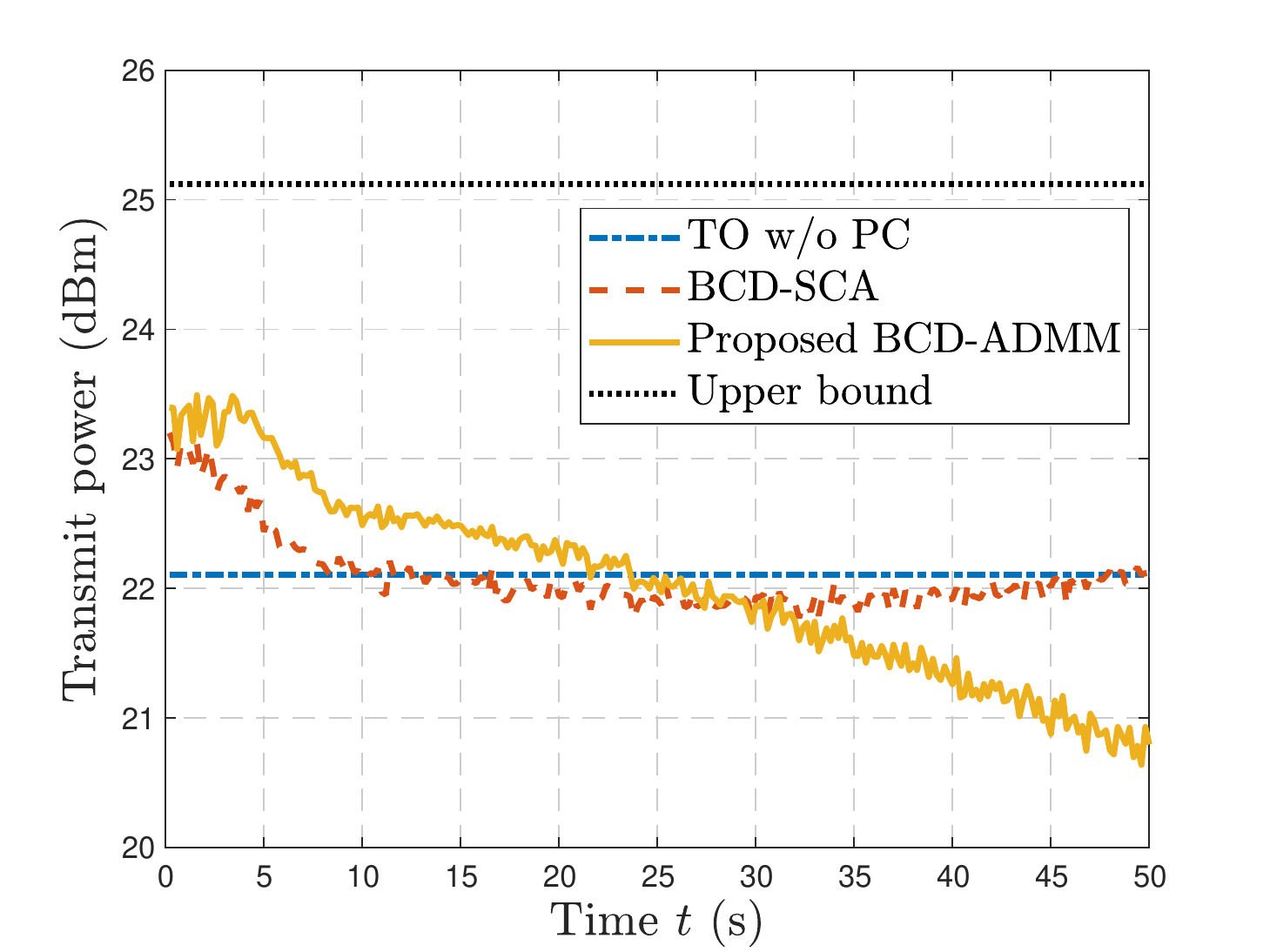}
		\vspace{-10mm}
	\caption{Total transmit power consumption for different algorithms. }\label{fig:tr_SP}
		\vspace{-8mm}
	\end{minipage}
\end{figure}

%\begin{figure}[t]
%	\centering
%	\includegraphics[width=8cm,height=6cm]{Figs/Traj}
%	\vspace{-1mm}
%	\caption{UAV trajectories for different algorithms when $T = 50$~s. }\label{fig:Traj}
%	\vspace{-1mm}
%\end{figure}
%
%\begin{figure}[t]
%	\centering
%	\includegraphics[width=8cm,height=6cm]{Figs/Sp}
%	\vspace{-1mm}
%	\caption{Total transmit power of all sensors for different algorithms. }\label{fig:tr_SP}
%	\vspace{-1mm}
%\end{figure}

Fig. \ref{fig:Traj} shows the trajectories of the UAV by applying different methods when mission duration $T = 50$ s. 
Each trajectory is sampled every three seconds.
%We consider a heterogeneous network, where $K=50$  sensors are separated into two clusters, i.e., $A$ with 15 sensors and $B$ with 35 sensors. 
%Instead of imposing a certain mobility model on individual sensors, we assign a different trace for each cluster.
In this experiment,  cluster $A$ moves with a constant speed of $5$ m/s and at an angle of ${\pi}/{2}$, while cluster $B$ is moving with a constant speed of $5$ m/s and at an angle of ${2\pi}/{3}$.
%The moving speeds of the centers of clusters $ A $ and $ B $ are $5$ m/s.
%The sensors in cluster $A$ and cluster $B$ are randomly and uniformly distributed in a circle at corresponding centers with a radius of $50$ meters, respectively.
Therein, navy blue solid lines with $\triangle$ represent the traces of the cluster centers (along which the sensor distribution is updated).
The  locations of sensors at $t = 0$~s and at $t = 50$~s are  marked by blue $\square$. The initial location of the UAV is marked by red  $\star$.
 
It is observed that for all algorithms, the UAV flies along an arc path.
The main reason for this arc path is that,  all links' channel conditions depend on the UAV's location at each time slot. 
The closer the UAV flies to one particular sensor, the farther it is away from some other sensors in general. 
It is also observed that the trajectories generated by the proposed BCD-ADMM algorithm, the benchmark TO w/o PC, and the BCD-SCA algorithm differ significantly. 
To unveil the difference in trajectories obtained by these algorithms and the reasons for such difference, we plot the corresponding sum transmit power of sensors over time slots, as shown in Fig. \ref{fig:tr_SP}. 
We observe that the proposed BCD-ADMM algorithm renders the UAV to move sufficiently close to the sensors to save their transmit power.
Furthermore,  the optimized trajectory is closer to cluster $B$ with a lower sensor transmit power budget ($P_B = 7$ dBm) as compared to cluster $A$ ($P_A = 10$ dBm).
This demonstrates that, for the case of unequal transmit power budgets, the proposed BCD-ADMM algorithm can efficiently strike a balance between minimizing communication distance and the sensors' transmit power, thereby enhancing the performance of AirComp.
However, for the benchmark TO w/o PC with constant transmit power, the UAV cannot perfectly track the sensors' mobility, resulting in a less effective trajectory.
Additionally, for $T<25$~s, the BCD-SCA algorithm decreases the sensors' transmit power, while the optimized trajectory of UAV is similar to that of BCD-ADMM.
However, when $T>25$~s, different from the BCD-ADMM algorithm, 
the BCD-SCA algorithm opts to increase the sensors' transmit power instead of shortening distances between the sensors and the UAV.
As a result, the BCD-SCA algorithm has to use a higher transmit power as compared to the BCD-ADMM algorithm.

\begin{figure}[t]
	\centering
	\begin{minipage}{.48\textwidth}
		\centering
		\includegraphics[width=8cm,height=6cm]{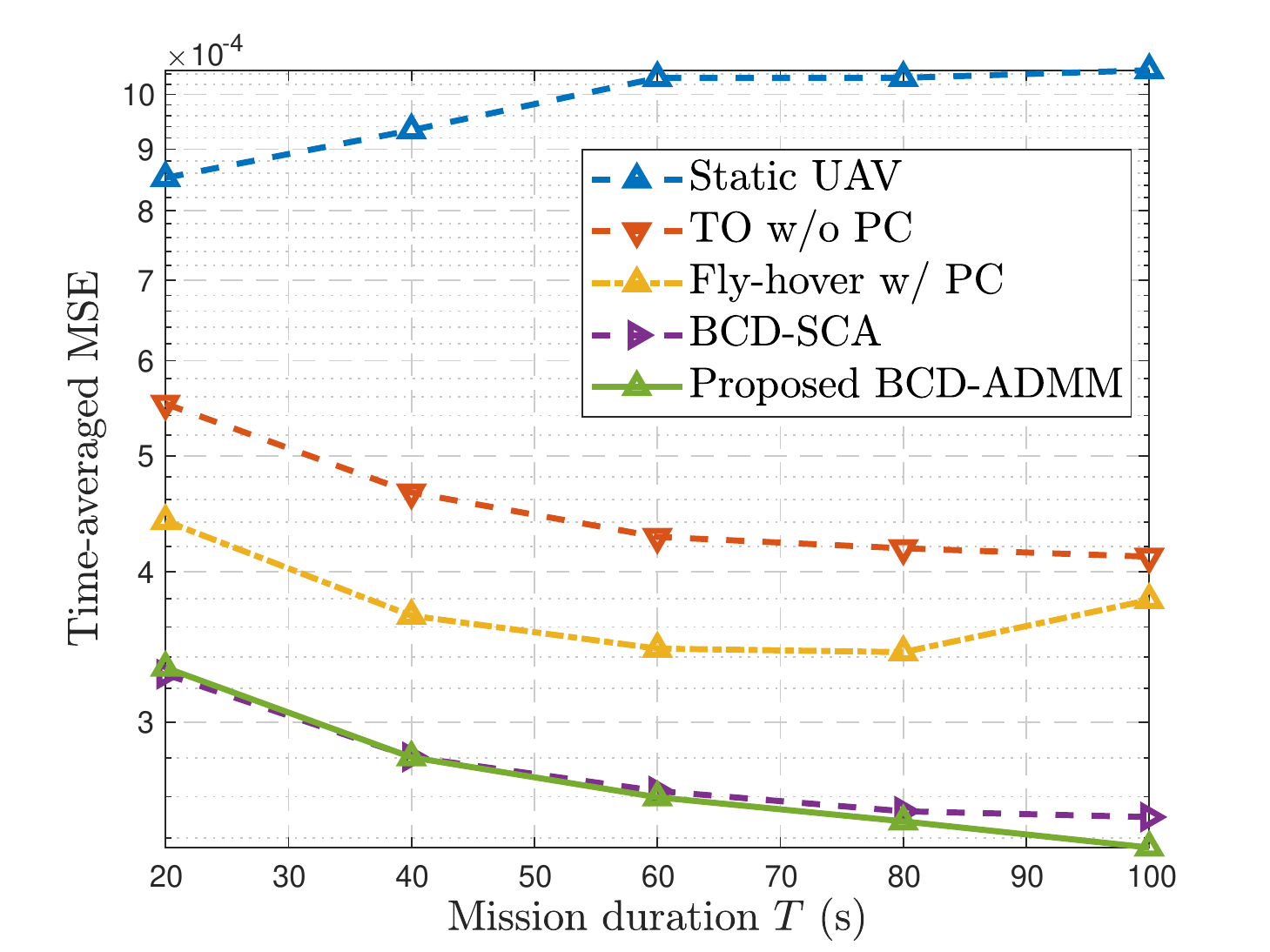}
		\vspace{-10mm}
		\caption{Time-averaged MSE versus mission duration $T$. }\label{fig:AlogT}
		\vspace{-8mm}
	\end{minipage}
	\hspace{1mm}
	\begin{minipage}{.48\textwidth}
		\centering
		\includegraphics[width=8cm,height=6cm]{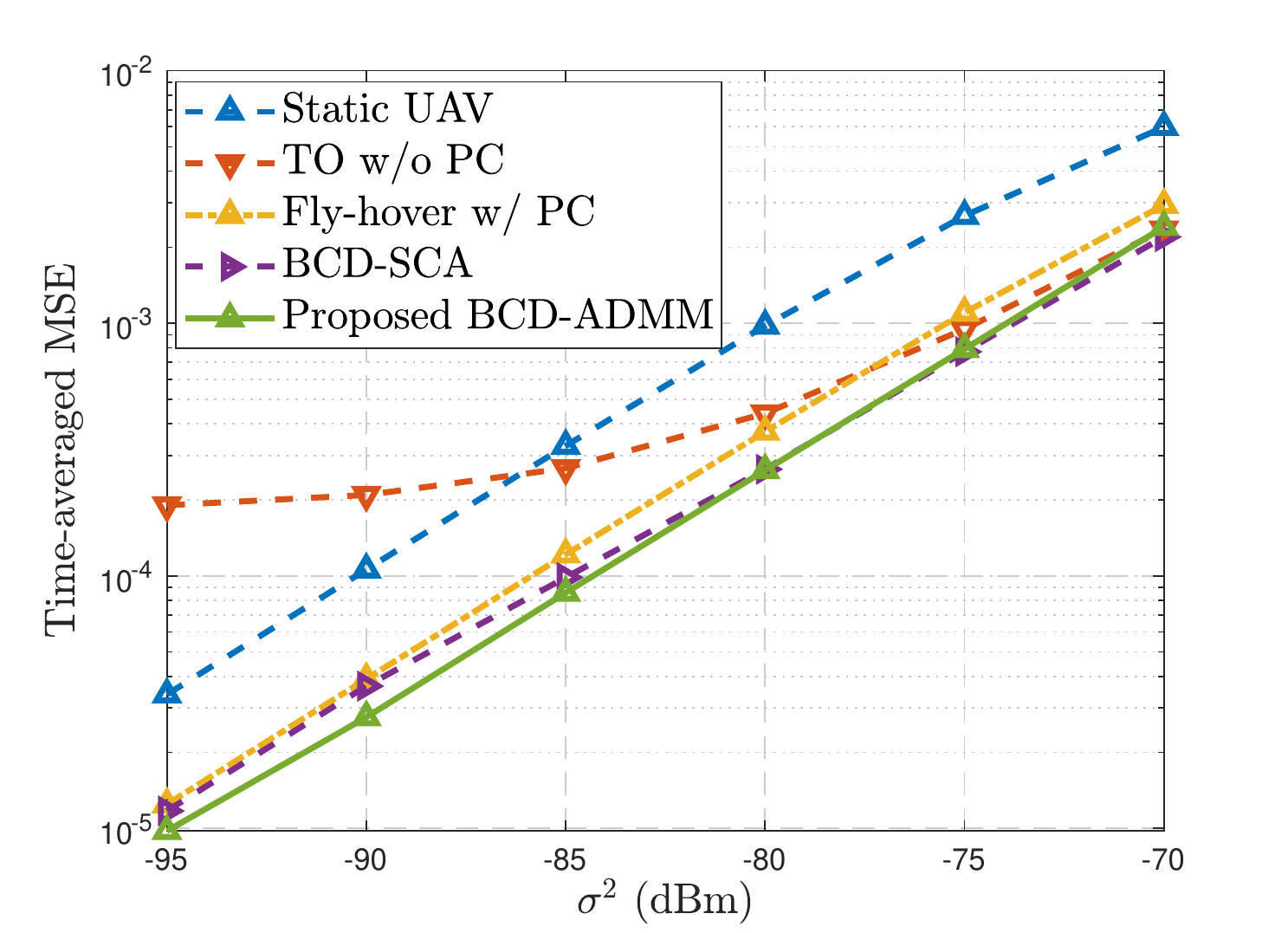}
		\vspace{-10mm}
		\caption{Time-averaged MSE versus noise power $\sigma^2$. }\label{fig:AlogNoise}
		\vspace{-8mm}
	\end{minipage}
\end{figure}

\subsection{Performance Comparison of Different Algorithms}
%\begin{figure}[t]
%	\centering
%	\includegraphics[width=8cm,height=6cm]{Figs/AlogT}
%	\vspace{-1mm}
%	\caption{Time-averaged MSE versus mission duration $T$. }\label{fig:AlogT}
%	\vspace{-1mm}
%\end{figure}
%\begin{figure}[t]
%	\centering
%	\includegraphics[width=8cm,height=6cm]{Figs/AlogNoise}
%	\vspace{-1mm}
%	\caption{Time-averaged MSE versus
%		noise power $\sigma^2$. }\label{fig:AlogNoise}
%	\vspace{-1mm}
%\end{figure}
Fig. \ref{fig:AlogT} shows the  time-averaged MSE of different algorithms  versus different values of mission duration $T$. 
It is observed that the time-averaged MSE achieved by the static UAV scheme increases with $T$.
This is because the sensors move  farther away from their initial positions as $T$ increases, which in turn leads to worse  channel conditions.
In contrast, as shown in other four curves, by exploiting UAV mobility to track the movement of sensors, the time-averaged MSE  is reduced compared to the static UAV scenario.
Furthermore, the joint design schemes (i.e., BCD-SCA and BCD-ADMM) have a smaller time-averaged MSE than  the TO w/o PC algorithm.
This is because these joint design schemes strike a better balance between minimizing links' distances and sensors' transmit power by fully exploiting the synergy of trajectory design and power control, while the inefficient usage of transmit power in the TO w/o PC algorithm results in less efficient trajectory of the UAV which in turn degrades the performance of AirComp.
%This demonstrates that the UAV trajectory should be jointly designed with power control.
It can also be observed that the joint design schemes outperform the heuristic trajectory algorithm (i.e., Fly-hover w/ PC), especially when the mission duration is long.
All the above results illustrate the importance and necessity of the joint design in minimizing the time-averaged MSE for AirComp.
Additionally, it is found that the proposed BCD-ADMM algorithm further achieves  performance improvement. 
This is because  the proposed method obtains the optimal solution of each subproblem, while the BCD-SCA method only optimizes the approximate lower bound of the trajectory problem based on the SCA technique.

Fig. \ref{fig:AlogNoise} shows the robustness of different algorithms against noise  power when $T = 50$~s.
It is observed that the time-averaged MSE achieved by all algorithms rises up with the noise power since  transmit powers at the sensors are limited.
Interestingly, in a relatively lower noise power region, the TO w/o PC algorithm achieves the worst performance compared to other four algorithms with power control.
This is because when the noise power is small, a smaller normalizing factor is required for suppressing the noise-induced error, as demonstrated in \textbf{Remark \ref{eta remark}}  of Section III. 
Consequently, from \textbf{Proposition \ref{theta solution lemma}}, the required power for aligning the signals will be varied with normalizing factors. 
In contrast,  using a constant power strategy will cause a increased signal misalignment error.
These results suggest that for the case with a  relatively lower noise power, it is important to adopt a power control strategy for reducing the time-averaged MSE.
What's more, we can observe that the proposed algorithm outperforms other four benchmarks for different noise power values. 
This demonstrates that the proposed algorithm is robust to noise power variations.

\section{Conclusion} \label{C section}
We studied the time-averaged MSE minimization problem in  a UAV-aided AirComp system with mobile sensors, taking into account the UAV trajectory design, receive normalizing factors optimization at the UAV, and transmit power control at the sensors.
By introducing a novel variable transformation and applying the BCD technique, the equivalently reformulated problem can be reduced to a convex subproblem.
We derived the optimal closed-form expressions for intermediate variables and normalizing factors.
Furthermore,  the convex QCQP subproblem of trajectory design was reformulated in an ADMM form to reduce the computational complexity.
Simulation results demonstrated the superiority of the proposed low-complexity algorithm in minimizing time-averaged MSE and reducing the simulation time compared to the existing algorithms.
This initial investigation demonstrated the effectiveness of
deploying an UAV to assist AirComp for data aggregation applications.
For future studies, the joint design framework developed in this paper can be extended to more general scenarios with multiple cooperative UAVs and in the presence of ground BSs, while taking into account the prediction errors on the movement of sensors for practical implementation.

\appendix 
\vspace{-1mm}
\subsection{Proof of Proposition \ref{real proposition}}\label{real proposition proof}
For each $n\in \mathcal{N}$, $k\in \mathcal{K}$, each term $\big({b_k[n]h_{k}[n]}/{\eta[n]}-1\big)^2$ in \eqref{MSE_s_y_t} follows that
\begin{eqnarray}
\Bigg(\frac{b_k[n]h_{k}[n]}{\eta[n]}-1\Bigg)^2 &=& \Big|\frac{b_k[n]h_{k}[n]}{\eta[n]}\Big|^2+1-2\mathcal{R}\Big(\frac{b_k[n]h_{k}[n]}{\eta[n]}\Big)\nonumber\\
%&\!\!\!\!=& |b_k[n]|^2|h_{k}[n]|^2/|\eta[n]|^2+1-2\mathcal{R}\Big(\frac{b_k[n]h_{k}[n]}{\eta[n]}\Big)\nonumber\\
&\geq &|b_k[n]|^2|h_{k}[n]|^2/|\eta[n]|^2+1-2|b_k[n]||h_{k}[n]|/|\eta[n]|,
\end{eqnarray}
where the equality holds only when ${b_k[n]h_{k}[n]}/{\eta[n]}$ is real and non-negative.
Therefore, with any given  amplitudes of  $\{b_k[n]\}$, $\{\eta[n]\}$, and  $\{h_k[n]\}$, the $\overline{\sf MSE}$ attains the minimum only when each term ${b_k[n]h_{k}[n]}/{\eta[n]}, \forall n, \forall k$ is real and non-negative. This completes the proof.

\vspace{-4mm}
\subsection{Proof of Proposition \ref{theta solution lemma}}\label{theta solution lemma proof}
By setting the first derivative of 	$\mathcal{L}(\{\theta_k[n]\},  \{\alpha_n\}, \lambda)$ w.r.t. $\theta_k[n]$ to zero as follows
\begin{eqnarray}
\frac{\partial \mathcal{L}}{\partial{\theta_k[n]}} = \frac{1}{\eta^2[n]} + \frac{\alpha_n + \lambda}{|h_k[n]|^2} -  \frac{1}{\sqrt{\theta_k[n]}\eta[n]}
 =0,
\end{eqnarray}
we obtain
\vspace{-2mm}
\begin{eqnarray}\label{derivative solution}
\theta_k[n]=  \Big(\frac{\eta[n]|h_k[n]|^2}{|h_k[n]|^2 + (\alpha_n +  \lambda)\eta^2[n]}\Big)^2.
\end{eqnarray}

If  $\alpha_n > 0$ holds, the peak power constraint of sensor $k$ at time slot $n$ must be tight at the optimality due to the complementary slackness condition, i.e.,
\begin{eqnarray}\label{complementary}
\alpha_n\left(\frac{\theta_k[n]}{|h_k[n]|^2} - P_k \right) =0.
\end{eqnarray}
Thus, from \eqref{derivative solution} and \eqref{complementary},  we obtain
\begin{eqnarray}
\!\!\theta^{\star}_k[n]     =   \Big(\frac{\eta[n]|h_k[n]|^2}{|h_k[n]|^2 + (\alpha_n +  \lambda)\eta^2[n]}\Big)^2 =  P_k|h_k[n]|^2.
\end{eqnarray}
Furthermore, when $\alpha_n > 0$, it is easily verified that 
\begin{eqnarray}
\Big(\frac{\eta[n]|h_k[n]|^2}{|h_k[n]|^2 + (\alpha_n +  \lambda)\eta^2[n]}\!\Big)^2 < \Big(\frac{\eta[n]|h_k[n]|^2}{|h_k[n]|^2 +   \lambda\eta^2[n]}\Big)^2 < \eta^2[n].
\end{eqnarray}
Therefore, when $\alpha_n > 0$ holds,  we have 
\begin{eqnarray}\label{geq}
 \theta^{\star}_k[n] =  P_k|h_k[n]|^2   <  \Big(\frac{\eta[n]|h_k[n]|^2}{|h_k[n]|^2 +  \lambda\eta^2[n]}\Big)^2 < \eta^2[n].
\end{eqnarray}
% ------------------------------------------------------------
While if  $\alpha_n = 0$ holds, by substituting $\alpha_n = 0$ into \eqref{derivative solution}, we obtain that
\begin{eqnarray} \label{equal}
\theta^{\star}_k[n] =  
\left\{
\begin{aligned}
& \eta^2[n],
 \ \  \text{if} \ \  \lambda = 0   \\
&\Big(\frac{\eta[n]|h_k[n]|^2}{|h_k[n]|^2 +  \lambda\eta^2[n]}\Big)^2,   \text{otherwise}.
\end{aligned}
\right.
\end{eqnarray}
Furthermore, if $\alpha_n = 0$ holds, due to the primal feasibility condition, the peak power constraint of sensor $k$ at time slot $n$ must be satisfied with the following inequality at the optimality, 
\begin{eqnarray}\label{primal feasibility2}
\left(\frac{\theta^{\star}_k[n]}{|h_k[n]|^2} - P_k \right) \leq 0 \Leftrightarrow \theta^{\star}_k[n]\leq P_k|h_k[n]|^2.
\end{eqnarray}
To sum up, from \eqref{geq}, \eqref{equal}, and \eqref{primal feasibility2}, the optimal $\theta^{\star}_k[n]$ is given by
\begin{eqnarray} 
\theta^{\star}_k[n] = 
 \left\{
\begin{aligned}
& \min \big\{ \eta^2[n], P_k\big|h_{k}[n]\big|^2 \big\},
 \ \ \ \text{if} \ \  \lambda = 0   \\
&\min\Big\{\Big(\frac{\eta[n]|h_{k}[n]|^2}{|h_{k}[n]|^2 + \lambda^{\star}\eta^2[n]}\Big)^2, P_k|h_{k}[n]|^2 \Big\},   \text{otherwise},
\end{aligned}
\right.
\end{eqnarray}
where  $\lambda^{\star}$ is a constant that ensures the average power constraint $\sum_{n=1}^N {\theta^{\star}_k[n]}/{|h_k[n]|^2} =  N\bar{P}_k$ to satisfy the  complementary slackness condition when $\lambda >0$, i.e.,
\begin{eqnarray}
\lambda\left(\sum_{n=1}^{N}\frac{\theta_k[n]}{|h_k[n]|^2} - NP_k \right) =0.
\end{eqnarray}
Furthermore, if $\lambda = 0$ holds, due to the primal feasibility condition, the average power constraint of sensor $k$  must be satisfied with the following inequality at the optimality, 
\begin{eqnarray}\label{primal feasibility1}
\sum_{n=1}^{N}\!\frac{\theta^{\star}_k[n]}{|h_k[n]|^2} - N\bar{P}_k  \leq 0 \Leftrightarrow \sum_{n=1}^{N}\min\Big\{ \frac{\eta^2[n]}{\big|h_{k}[n]\big|^2}, P_k \Big\}\leq N\bar{P}_k.
\end{eqnarray}
In summary,  the optimal solution  of problem \eqref{k-th theta} is
\begin{eqnarray*} 
&&\theta^{\star}_k[n] = 
\left\{
\begin{aligned}
& \min\Big\{ \eta^2[n], P_k|h_{k}[n]|^2 \!\Big\}, 
\text{if} \  \min\Big\{ \frac{\eta^2[n]}{|h_{k}[n]|^2}, P_k \Big\}\leq N\bar{P}_k,  \\
&\min\Big\{ \Big(\frac{\eta[n]|h_{k}[n]|^2}{|h_{k}[n]\big|^2 +  \lambda^{\star}\eta^2[n]}\Big)^2, P_k|h_{k}[n]|^2 \Big\},   \text{otherwise}.
\end{aligned}
\right.
\end{eqnarray*}
This thus completes the proof.

\vspace{-4mm}
\subsection{Proof of Proposition \ref{convergence proposition}}\label{convergence proposition proof}
We denote  $f\big(\bm \eta, \bm \theta, \bm q \big)$  as the objective value of $\mathscr{P}$ for a feasible solution  $\big(\bm \eta, \bm \theta, \bm q \big)$.
As shown in step 7 of Algorithm \ref{algo1}, a feasible solution of problem $\mathscr{P}_{1.3}$ (i.e., $\big( \bm \eta^{i}, \bm \theta^{i}, \bm q^{i} \big)$) is also feasible to problem $\mathscr{P}_{1.1}$ and problem $\mathscr{P}_{1.2}$.
We denote $\big( \bm \eta^{i}, \bm \theta^{i}, \bm q^{i} \big)$ and $\big( \bm \eta^{i+1}, \bm \theta^{i+1}, \bm q^{i+1} \big)$ as a  feasible solution of   $\mathscr{P}$ at the $i$-th and $(i+1)$-th iterations, respectively. 

Since for given $\bm \theta^{i}, \bm q^{i}$ as  shown in step 5 of Algorithm \ref{algo1}, $\bm \eta^{i+1}$ is the optimal solution to problem $\mathscr{P}_{1.1}$, we have
\vspace{-2mm}
\begin{eqnarray} \label{p1}
f\big( \bm \eta^{i}, \bm \theta^{i}, \bm q^{i} \big) \geq f\big( \bm \eta^{i+1}, \bm \theta^{i}, \bm q^{i} \big). 	
\end{eqnarray}
Similarly, since for given $\bm \eta^{i+1}, \bm q^{i}$ as  shown in step 6 of Algorithm \ref{algo1}, $\bm \theta^{i+1}$ is the optimal solution to problem $\mathscr{P}_{1.2}$, it follows that
\vspace{-2mm}
\begin{eqnarray}\label{p2}
f\big( \bm \eta^{i+1}, \bm \theta^{i}, \bm q^{i} \big) \geq f\big( \bm \eta^{i+1}, \bm \theta^{i+1}, \bm q^{i} \big). 	
\end{eqnarray}
Besides, we have
\vspace{-2mm}
\begin{eqnarray}\label{p3}
f\big( \bm \eta^{i+1}, \bm \theta^{i+1}, \bm q^{i+1} \big) = f\big( \bm \eta^{i+1}, \bm \theta^{i+1}, \bm q^{i} \big). 	
\end{eqnarray}
This holds because the original objective function $f$ is independent of $\bm q$ but depends on $\bm \eta$ and $\bm \theta$.
Based on $\eqref{p1}$, $\eqref{p2}$, and $\eqref{p3}$, we further obtain 
\begin{eqnarray}
f\big( \bm \eta^{i+1}, \bm \theta^{i+1}, \bm q^{i+1} \big) \leq f\big( \bm \eta^{i}, \bm \theta^{i}, \bm q^{i} \big), 	
\end{eqnarray}
which shows that the objective value of problem $\mathscr{P}$ is always decreasing over iterations. Therefore, the proposed BCD-ADMM algorithm converges. This thus completes the proof.

\bibliographystyle{IEEEtran}
\bibliography{ref} % BM

\end{document}